\pdfoutput=1
\documentclass[12pt]{article}

\usepackage{amsmath,amssymb,amsthm,mathtools}
\usepackage{tcolorbox}
\usepackage{fullpage}
\usepackage{authblk}

\usepackage[utf8]{inputenc}
\usepackage[T1]{fontenc}

\usepackage[style=alphabetic,natbib=true,sortcites=true,backref=true,maxnames=99,maxalphanames=99,maxcitenames=99]{biblatex}

\DefineBibliographyStrings{english}{%
  backrefpage = {$\Lsh$~p.},
  backrefpages = {$\Lsh$~pp.},
}
\usepackage[colorlinks=true,citecolor=blue!65!black,linkcolor=red!75!black]{hyperref}
\usepackage[nameinlink]{cleveref}

\usepackage{booktabs,enumitem}
\usepackage{subfig}
\usepackage{nicefrac}
\usepackage{xfrac}

\usepackage{kbordermatrix}
\renewcommand{\kbldelim}{[}
\renewcommand{\kbrdelim}{]}

\usepackage{comment}
\usepackage{url}

\usepackage{xspace}
\usepackage{bm}

\usepackage[nameinlink]{cleveref}

\usepackage{lineno}

\crefname{theorem}{Theorem}{Theorems}
\crefname{lemma}{Lemma}{Lemmas}
\crefname{claim}{Claim}{Claims}
\crefname{corollary}{Corollary}{Corollaries}
\crefname{remark}{Remark}{Remarks}
\crefname{observation}{Observation}{Observations}
\crefname{hypothesis}{Hypothesis}{Hypotheses}
\crefname{definition}{Definition}{Definitions}
\crefname{problem}{Problem}{Problems}

\crefname{appendix}{Appendix}{Appendices}
\crefname{section}{Section}{Sections}
\crefname{equation}{Eq.}{Eqs.}
\crefname{figure}{Figure}{Figures}
\crefname{table}{Table}{Tables}

\renewcommand{\geq}{\geqslant}
\renewcommand{\leq}{\leqslant}
\renewcommand{\epsilon}{\varepsilon}

\newcommand{\prb}[1]{\textup{\textsc{#1}}\xspace}

\renewcommand{\vec}[1]{\mathbf{\bm{#1}}}
\newcommand{\mat}[1]{\mathbf{\bm{#1}}}
\newcommand{\abs}[1]{\vert{#1}\vert}

\newcommand{\rme}{\mathrm{e}}
\newcommand{\OPT}{\mathsf{OPT}}
\newcommand{\hati}{\hat{\imath}}
\newcommand{\hatj}{\hat{\jmath}}
\newcommand{\isep}{\mathrel{..}\nobreak}
\newcommand{\haba}{\Delta}

\DeclareMathOperator{\bigO}{\mathcal{O}}
\DeclareMathOperator{\argmax}{argmax}
\DeclareMathOperator{\vol}{vol}
\DeclareMathOperator{\proj}{proj}
\DeclareMathOperator{\cons}{cons}
\DeclareMathOperator{\cov}{cov}
\DeclareMathOperator{\dup}{dup}
\DeclareMathOperator{\dis}{d}
\DeclareMathOperator{\sgn}{sgn}
\DeclareMathOperator{\rank}{rank}
\DeclareMathOperator{\nz}{nz}
\DeclareMathOperator{\opt}{opt}
\DeclareMathOperator{\maxdet}{maxdet}
\let\det\relax\DeclareMathOperator{\det}{det}

\newcommand{\cP}{\textup{\textsf{P}}\xspace}
\newcommand{\NP}{\textup{\textsf{NP}}\xspace}
\newcommand{\shP}{$\#$\textsf{P}\xspace}
\newcommand{\FPT}{\textup{\textsf{FPT}}\xspace}
\newcommand{\W}[1]{\textup{\textsf{W}[#1]}\xspace}
\newcommand{\XP}{\textup{\textsf{XP}}\xspace}

\newcommand{\calC}{\mathcal{C}}

\newcommand{\calF}{\mathcal{F}}

\newcommand{\calI}{\mathcal{I}}

\newcommand{\calS}{\mathcal{S}}

\newcommand{\bbN}{\mathbb{N}}

\newcommand{\bbQ}{\mathbb{Q}}
\newcommand{\bbR}{\mathbb{R}}

\usepackage[all]{nowidow}

\newtheorem{theorem}{Theorem}[section]
\newtheorem*{theorem*}{Theorem}
\newtheorem{lemma}[theorem]{Lemma}
\newtheorem{corollary}[theorem]{Corollary}

\newtheorem{claim}[theorem]{Claim}
\newtheorem*{claim*}{Claim}

\newtheorem{observation}[theorem]{Observation}
\newtheorem{hypothesis}[theorem]{Hypothesis}
\theoremstyle{definition}

\newtheorem{problem}[theorem]{Problem}

\newenvironment{claimproof}{\begin{proof}}{\end{proof}}

\numberwithin{equation}{section}

\title{On the Parameterized Intractability of \\ Determinant Maximization\footnote{
A preliminary version of this paper appeared in \emph{Proc.~33rd Int.~Symp.~on Algorithms and Computation (ISAAC)}, 2022 \cite{ohsaka2022parameterized}.
}}
\author{Naoto Ohsaka\thanks{
CyberAgent, Inc., Tokyo, Japan.
\href{mailto:ohsaka\_naoto@cyberagent.co.jp}{\texttt{ohsaka\_naoto@cyberagent.co.jp}};
\href{mailto:naoto.ohsaka@gmail.com}{\texttt{naoto.ohsaka@gmail.com}}
}}
\date{\today}

\begin{document}
\maketitle

\thispagestyle{empty}
\begin{abstract}\noindent In the \prb{Determinant Maximization} problem,
given an $n \times n$ positive semi-definite matrix $\mat{A}$ in $\bbQ^{n \times n}$ and an integer $k$,
we are required to find a $k \times k$ principal submatrix of $\mat{A}$ having the maximum determinant.
This problem is known to be \NP-hard and
further proven to be \W{1}-hard with respect to $k$ by {Koutis} 
(Inf.~Process.~Lett., 2006)
\cite{koutis2006parameterized}; i.e.,
a $f(k)n^{\bigO(1)}$-time algorithm is unlikely to exist
for any computable function~$f$.
However, there is still room to explore its parameterized complexity in the \emph{restricted case},
in the hope of overcoming the general-case parameterized intractability.
In this study, we rule out the fixed-parameter tractability of
\prb{Determinant Maximization} even if 
an input matrix is extremely sparse or low rank, or
an approximate solution is acceptable.
We first prove that \prb{Determinant Maximization} is
\NP-hard and \W{1}-hard even if an input matrix is an \emph{arrowhead matrix};
i.e., the underlying graph formed by nonzero entries is a star,
implying that the structural sparsity is not helpful.
By contrast,
\prb{Determinant Maximization} is known to be solvable in polynomial time on \emph{tridiagonal matrices} \cite{al-thani2021tridiagonal}.
Thereafter, we demonstrate
the \W{1}-hardness with respect to the \emph{rank} $r$ of an input matrix.
Our result is stronger than Koutis' result in the sense that
any $k \times k$ principal submatrix is singular whenever $k > r$.
We finally give evidence that
it is \W{1}-hard to approximate \prb{Determinant Maximization}
parameterized by $k$ within a factor of $2^{-c\sqrt{k}}$ for some universal constant $c > 0$.
Our hardness result is conditional on the \emph{Parameterized Inapproximability Hypothesis}
posed by
{Lokshtanov, Ramanujan, Saurab, and Zehavi} (SODA~2020) \cite{lokshtanov2020parameterized},
which asserts that
a gap version of \prb{Binary Constraint Satisfaction Problem} is \W{1}-hard.
To complement this result,
we develop an $\epsilon$-additive approximation algorithm that 
runs in $\epsilon^{-r^2} \cdot r^{\bigO(r^3)} \cdot n^{\bigO(1)}$ time
for the rank $r$ of an input matrix,
provided that the diagonal entries are bounded.
\end{abstract}
\clearpage
\tableofcontents
\clearpage
\begin{figure}
    \null\hfill
    \subfloat[
        Four vectors $\vec{v}_1=(5,0,0)$, $\vec{v}_2=(2,3,0)$,
        $\vec{v}_3=(1,1,3)$, and $\vec{v}_4=(3,1,1)$.
    ]{\includegraphics[scale=0.7]{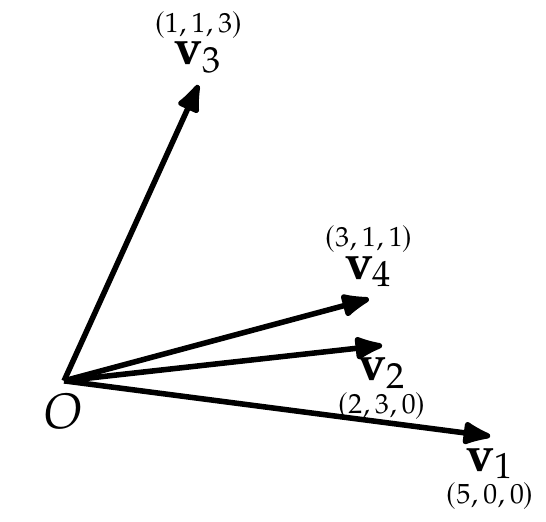}}
    \hfill
    \subfloat[Gram matrix $\mat{A}$ of $\vec{v}_1,\vec{v}_2,\vec{v}_3,\vec{v}_4$.]{
    $\mat{A} = \kbordermatrix{
          &  1 &  2 &  3 &  4 \\
        1 & 25 & 10 &  5 & 15 \\
        2 & 10 & 13 &  5 &  9 \\
        3 &  5 &  5 & 11 &  7 \\
        4 & 15 &  9 &  7 & 11 
    }$
    }
    \hfill\null
    \\
    \null\hfill
    \subfloat[
        Parallelepiped spanned by $\vec{v}_1, \vec{v}_2, \vec{v}_3$, whose determinant is $2{,}025=45^2$.
    ]{\includegraphics[scale=0.55]{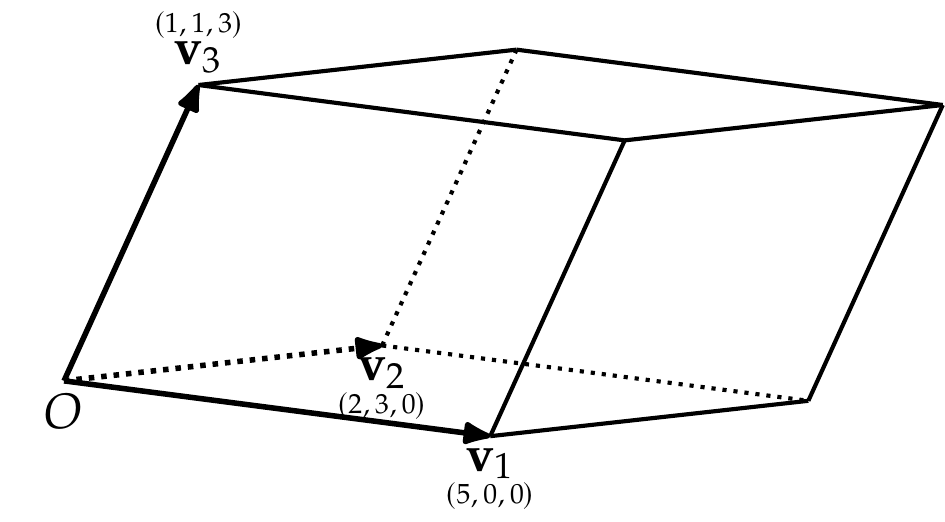}}
    \hfill
    \subfloat[
        Parallelepiped spanned by $\vec{v}_1, \vec{v}_2, \vec{v}_4$,
        whose determinant is $225=15^2$, which is smaller than $2{,}025$, even though $\|\vec{v}_3\| = \|\vec{v}_4\|.$
    ]{\includegraphics[scale=0.55]{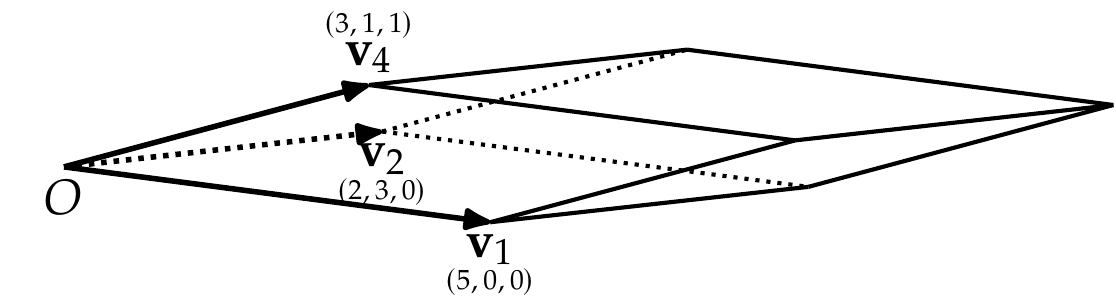}}
    \hfill\null
    \caption{Example of \prb{Determinant Maximization} with $n=4$ and $k=3$.}
    \label{fig:example}
\end{figure}

\section{Introduction}\label{sec:intro}
\subsection{Background}
We study the following \prb{Determinant Maximization} problem:
Given an $n \times n$ positive semi-definite matrix
$\mat{A}$ in $\bbQ^{n \times n}$ and an integer $k $ in $[n]$ denoting the solution size,
find a $k \times k$ principal submatrix of $\mat{A}$ having
the maximum determinant;
namely, maximize $\det(\mat{A}_{S})$ subject to $S \in {[n] \choose k}$.
One motivating example for this problem is a \emph{subset selection task}.
Suppose we are given $n$ items (e.g., images or products) 
associated with feature vectors
$\vec{v}_1, \ldots, \vec{v}_n$ and
required to select a ``diverse'' set of $k$ items among them.
We can measure the diversity of a set $S$ of $k$ items using
the principal minor $\det(\mat{A}_S)$ of the Gram matrix $\mat{A}$
defined by feature vectors such that
$A_{i,j} \triangleq \langle \vec{v}_i, \vec{v}_j \rangle$ for all
$i, j \in [n]$,
resulting in \prb{Determinant Maximization}.
This formulation is justified by the fact that
$\det(\mat{A}_S)$ is equal to the squared volume of the parallelepiped
spanned by $\{\vec{v}_i : i \in S\}$; that is,
a pair of vectors at a large angle is regarded as more diverse.
See \cref{fig:example} for an example of \prb{Determinant Maximization} and its volume interpretation.
In artificial intelligence and machine learning communities,
\prb{Determinant Maximization} is also known as MAP inference on
a \emph{determinantal point process} \cite{borodin2005eynard,macchi1975coincidence},
and has found many applications over the past decade, 
including 
tweet timeline generation \cite{yao2016tweet},
object detection \cite{lee2016individualness},
change-point detection \cite{zhang2016block},
document summarization \cite{chao2015large,kulesza2011kdpps},
YouTube video recommendation \cite{wilhelm2018practical}, and
active learning \cite{biyik2019batch}.
See the survey of \citet{kulesza2012determinantal} for further details.
Though \prb{Determinant Maximization} is known to be \NP-hard to solve exactly \cite{ko1995exact},
we can achieve an $\rme^{-k}$-factor approximation in polynomial time \cite{nikolov2015randomized},
which is nearly optimal because
a $2^{-ck}$-factor approximation for some constant $c > 0$ is impossible unless \cP~$=$~\NP
\cite{civril2013exponential,koutis2006parameterized,summa2014largest}.

Having known a nearly tight hardness-of-approximation result in the polynomial-time regime,
we resort to \emph{parameterized algorithms}
\cite{flum2006parameterized,cygan2015parameterized,downey2012parameterized}.
We say that a problem is \emph{fixed-parameter tractable} (FPT)
with respect to a parameter $k \in \bbN$
if it can be solved in $f(k) \abs{\calI}^{\bigO(1)}$ time
for some computable function $f$ and instance size $\abs{\calI}$.
One very natural parameter is the \emph{solution size} $k$, which is expected to be small in practice.
By enumerating all $k \times k$ principal submatrices,
we can solve \prb{Determinant Maximization} in $n^{k + \bigO(1)}$ time; i.e.,
it belongs to the class \XP.
Because \FPT~$\subsetneq$~\XP \cite{downey2012parameterized},
it is even more desirable if an FPT algorithm exists.
Unfortunately, \citet{koutis2006parameterized} has already proven that
\prb{Determinant Maximization} is \W{1}-hard with respect to $k$,
which in fact follows from the reduction due to \citet{ko1995exact}.
Therefore, under the widely-believed assumption that \FPT~$\neq$~\W{1},
an FPT algorithm for \prb{Determinant Maximization} does not exist.

However, there is still room to explore the parameterized complexity of \prb{Determinant Maximization} in the \emph{restricted case},
in the hope of circumventing the general-case parameterized intractability.
Here, we describe three possible scenarios.
One can first assume an input matrix $\mat{A}$ to be \emph{sparse}.
Of particular interest is the structural sparsity of
the \emph{symmetrized graph} of $\mat{A}$ \cite{courcelle2001fixed,cifuentes2016efficient} defined as
the underlying graph formed by nonzero entries of $\mat{A}$,
encouraged by numerous FPT algorithms for \NP-hard graph-theoretic problems 
parameterized by the treewidth
\cite{cygan2015parameterized,fomin2010exact}.
For example, in change-point detection applications,
\citet{zhang2016block} observed a small-bandwidth matrix and developed an efficient heuristic for \prb{Determinant Maximization}.
In addition, one may adopt a strong parameter.
The \emph{rank} of an input matrix $\mat{A}$ is such a natural candidate.
We often assume that $\mat{A}$ is low-rank in applications; for instance,
the feature vectors $\vec{v}_i$ are inherently low-dimensional \cite{celis2018fair}
or the largest possible subset is significantly smaller than the ground set size $n$.
Since any $k \times k$ principal submatrix of $\mat{A}$ is singular whenever $k > \rank(\mat{A})$,
we can ensure that $k \leq \rank(\mat{A})$; namely,
parameterization by $\rank(\mat{A})$ is considered stronger than that by $k$.
Intriguingly, the partition function of product determinantal point processes is 
FPT with respect to rank while \shP-hard in general \cite{ohsaka2020intractability}.
The last possibility to be considered is \emph{FPT-approximability}.
Albeit \W{1}-hardness of \prb{Determinant Maximization} with parameter $k$,
it could be possible to obtain an approximate solution in FPT time.
It has been demonstrated that
several \W{1}-hard problems can be approximated in FPT time, such as
\prb{Partial Vertex Cover} and \prb{Minimum $k$-Median} \cite{harpeled2004coresets}
(refer to the survey of
\citet{marx2008parameterized} and
\citet{feldmann2020survey}).
One may thus envision the existence of a $1/\rho(k)$-factor FPT-approximation algorithm
for \prb{Determinant Maximization} for a \emph{small} function $\rho$.
Alas, we refute the above possibilities under a plausible assumption in parameterized complexity.

\subsection{Our Results}
We improve
the \W{1}-hardness of \prb{Determinant Maximization}
due to \citet{koutis2006parameterized} by
showing that it is still \W{1}-hard
even if an input matrix is extremely sparse or low rank, or
an approximate solution is acceptable,
along with some tractable cases.

We first prove that
\prb{Determinant Maximization}
is \NP-hard and \W{1}-hard with respect to $k$
even if the input matrix $\mat{A}$ is an \emph{arrowhead matrix} (\cref{thm:arrowhead}).
An arrowhead matrix is a square matrix that can
include nonzero entries only in
the first row, the first column, or the diagonal; i.e.,
its symmetrized graph is a star (cf.~\cref{fig:arrowhead}).
Our hardness result implies that
the ``structural sparsity'' of input matrices is not helpful;
in particular,
it follows from \cref{thm:arrowhead} that
this problem is \NP-hard even if
the treewidth, pathwidth, and vertex cover number of
the symmetrized graph are all $1$.
The proof is based on
a parameterized reduction from \prb{$k$-Sum},
which is a parameterized version of \prb{Subset Sum} known to be \W{1}-complete \cite{abboud2014losing,downey1995fixed}, and involves
a structural feature of the determinant of arrowhead matrices.
On the other hand,
\prb{Determinant Maximization} is known to be solvable in polynomial time on \emph{tridiagonal matrices} \cite{al-thani2021tridiagonal}, whose symmetrized graph is a path graph (cf.~\cref{fig:tridiagonal}).
Though an extended abstract of this paper appearing in ISAAC'22
includes a polynomial-time algorithm on tridiagonal matrices,
Jon Lee pointed out to us that
\citeauthor{al-thani2021tridiagonal} already proved the polynomial-time solvability
on tridiagonal matrices in LAGOS'21~\cite{al-thani2021tridiagonal} and
on spiders of bounded legs~\cite{al-thani2021tridiagonala}.
We thus omitted the proof for tridiagonal matrices from this article.

\begin{figure}
\null\hfill
\begin{minipage}{0.4\hsize}
    \centering
    $\kbordermatrix{
     &0&1&2&3&4&5&6&7&8 \\
    0&*&*&*&*&*&*&*&*&* \\
    1&*&*&0&0&0&0&0&0&0 \\
    2&*&0&*&0&0&0&0&0&0 \\
    3&*&0&0&*&0&0&0&0&0 \\
    4&*&0&0&0&*&0&0&0&0 \\
    5&*&0&0&0&0&*&0&0&0 \\
    6&*&0&0&0&0&0&*&0&0 \\
    7&*&0&0&0&0&0&0&*&0 \\
    8&*&0&0&0&0&0&0&0&* \\
    }$
\caption{Structure of arrowhead matrices, where ``$*$'' denotes nonzero entries.}
\label{fig:arrowhead}
\end{minipage}
\hfill
\begin{minipage}{0.4\hsize}
    \centering
    $\kbordermatrix{
     &0&1&2&3&4&5&6&7&8 \\
    0&*&*&0&0&0&0&0&0&0 \\
    1&*&*&*&0&0&0&0&0&0 \\
    2&0&*&*&*&0&0&0&0&0 \\
    3&0&0&*&*&*&0&0&0&0 \\
    4&0&0&0&*&*&*&0&0&0 \\
    5&0&0&0&0&*&*&*&0&0 \\
    6&0&0&0&0&0&*&*&*&0 \\
    7&0&0&0&0&0&0&*&*&* \\
    8&0&0&0&0&0&0&0&*&* \\
    }$
\caption{Structure of tridiagonal matrices, where ``$*$'' denotes nonzero entries.}
\label{fig:tridiagonal}
\end{minipage}
\hfill\null
\end{figure}

Thereafter, we demonstrate that \prb{Determinant Maximization} is \W{1}-hard when parameterized by the \emph{rank} of an input matrix (\cref{cor:rank}).
In fact,
we obtain the stronger result that
it is \W{1}-hard to determine whether
an input set of $n$ $d$-dimensional vectors includes
$k$ pairwise orthogonal vectors
when parameterized by $d$ (\cref{thm:rank}).
Unlike the proof of \cref{thm:arrowhead},
we are allowed to construct only a $f(k)$-dimensional vector in a parameterized reduction.
Therefore, we reduce from
a different \W{1}-complete problem called \prb{Grid Tiling} due to \citet{marx2007optimality,marx2012tight}.
In \prb{Grid Tiling},
we are given
$k^2$ nonempty sets of integer pairs arranged in a $k \times k$ grid, and the task is to select $k^2$ integer pairs such that
the vertical and horizontal neighbors agree respectively in the first and second coordinates (see \cref{def:grid-tiling} for the precise definition).
\prb{Grid Tiling} is favorable for our purpose because
the constraint consists of simple equalities, and
each cell is adjacent to (at most) four cells.
To express the consistency between adjacent cells using
only a $f(k)$-dimensional vector,
we exploit Pythagorean triples.
It is essential in \cref{thm:rank} that
the input vectors can include \emph{both} positive and negative entries
in a sense that we can find
$k$ $d$-dimensional nonnegative vectors that are pairwise orthogonal
in FPT time with respect to $d$ (\cref{obs:cp-rank}).

Our final contribution is to give evidence that
it is \W{1}-hard to determine whether the optimal value of 
\prb{Determinant Maximization}
is equal to $1$ or at most $2^{-c\sqrt{k}}$ for some universal constant $c > 0$; namely,
\prb{Determinant Maximization} is FPT-inapproximable within a factor of $2^{-c\sqrt{k}}$ (\cref{thm:inapprox}).
Our result is conditional on the \emph{Parameterized Inapproximability Hypothesis} (PIH),
which is a conjecture posed by \citet{lokshtanov2020parameterized} 
asserting that a gap version of \prb{Binary Constraint Satisfaction Problem}
is \W{1}-hard when parameterized by the number of variables.
PIH can be thought of as a parameterized analogue of 
the PCP theorem \cite{arora1998probabilistic,arora1998proof}; e.g.,
\citet{lokshtanov2020parameterized} show that
assuming PIH and \FPT~$\neq$~\W{1},
\prb{Directed Odd Cycle Transversal} does not admit
a $(1-\epsilon)$-factor FPT-approximation algorithm for some $\epsilon > 0$.
The proof of \cref{thm:inapprox} involves FPT-inapproximability of \prb{Grid Tiling} under PIH,
which is reminiscent of \citeauthor{marx2007optimality}'s work \cite{marx2007optimality} and might be of some independent interest.
Because we cannot achieve an exponential gap by
simply reusing the parameterized reduction from \prb{Grid Tiling} of the second hardness result (as inferred from \cref{obs:approx} below),
we apply a gadget invented by \citet{civril2013exponential} to
construct an $\bigO(k^2 n^2)$-dimensional vector
for each integer pair of a \prb{Grid Tiling} instance.
We further show that
the same kind of hardness result does \emph{not} hold when parameterized by the rank $r$ of an input matrix.
Specifically,
we develop an $\epsilon$-additive approximation algorithm
that runs in $\epsilon^{-r^2} \cdot r^{\bigO(r^3)} \cdot n^{\bigO(1)}$ time for any $\epsilon > 0$,
provided that the diagonal entries are bounded (\cref{obs:approx}).

\subsection{More Related Work}
\label{subsec:intro:related}
\prb{Determinant Maximization} is not only applied
in artificial intelligence and machine learning
but also in computational geometry \cite{gritzmann1995largest} and discrepancy theory;
refer to \citet{nikolov2015randomized} and references therein.
On the negative side,
\citet{ko1995exact} prove that 
\prb{Determinant Maximization} is \NP-hard, and
\citet{koutis2006parameterized} proves that 
it is further \W{1}-hard.
\NP-hardness of approximating \prb{Determinant Maximization}
has been investigated in \cite{koutis2006parameterized,civril2013exponential,summa2014largest,ohsaka2022some}.
On the algorithmic side,
a greedy algorithm achieves an approximation factor of $1/k!$ \cite{civril2009selecting}.
Subsequently, \citet{nikolov2015randomized} gives
an $\rme^{-k}$-factor approximation algorithm;
partition constraints \cite{nikolov2016maximizing} and matroid constraints \cite{madan2020maximizing} are also studied.
Several \shP-hard computation problems over matrices including 
permanents \cite{courcelle2001fixed,cifuentes2016efficient},
hyperdeterminants \cite{cifuentes2016efficient}, and 
partition functions of product determinantal point processes
\cite{ohsaka2020intractability}
are efficiently computable if 
the treewidth of the symmetrized graph or the matrix rank is bounded.

\section{Preliminaries}\label{sec:pre}

\subsection{Notations and Definitions}
For two integers $m,n \in \bbN$ with $m \leq n$,
let $ [n] \triangleq \{1, 2, \ldots, n\} $ and
$[m\isep n] \triangleq \{m,m+1,\ldots, n-1,n\}$.
For a finite set $S$ and an integer $k$, we write ${S \choose k}$ for the family of all size-$k$ subsets of $S$.
For a statement $P$, $[\![P]\!]$ is $1$ if $P$ is true, and $0$ otherwise.
The base of logarithms is $2$.
Matrices and vectors are written in bold letters, and scalars are unbold.
The Euclidean norm is denoted $\| \cdot \|$;
i.e., $\|\vec{v}\| \triangleq \sqrt{\sum_{i \in [d]}(v(i))^2} $ for a vector $\vec{v} \in \bbR^{d}$.
We use $\langle \cdot, \cdot \rangle$
for the standard inner product; i.e.,
$ \langle \vec{v}, \vec{w} \rangle \triangleq \sum_{i \in [d]} v(i) \cdot w(i) $ for two vectors $\vec{v}, \vec{w} \in \bbR^{d}$.
For an $n \times n$ matrix $\mat{A}$ and an index set $S \subseteq [n]$,
we use $\mat{A}_S$ to denote the principal submatrix of $\mat{A}$ whose
rows and columns are indexed by $S$.
For an $m \times n$ matrix $\mat{A}$,
the \emph{spectral norm} $\|\mat{A}\|_2$ is defined as
the square root of the maximum eigenvalue of $\mat{A}^\top \mat{A}$ and
the \emph{max norm} is defined as
$\|\mat{A}\|_{\max} \triangleq \max_{i,j} \abs{A_{i,j}} $.
It is well-known that
$\|\mat{A}\|_{\max} \leq \|\mat{A}\|_2 \leq \sqrt{mn} \cdot \|\mat{A}\|_{\max}$.
The \emph{symmetrized graph} \cite{courcelle2001fixed,cifuentes2016efficient} of an $n \times n$ matrix $\mat{A}$ 
is defined as an undirected graph $G$ that has
each integer of $[n]$ as a vertex and an edge $(i,j) \in {[n] \choose 2}$
if $A_{i,j} \neq 0$ or $A_{j,i} \neq 0$; i.e.,
$G = ([n], \{ (i,j) : A_{i,j} \neq 0 \})$.
For a matrix $\mat{A} \in \bbR^{n \times n}$,
its \emph{determinant} is defined as follows:
\begin{align}
    \det(\mat{A}) \triangleq \sum_{\sigma \in \mathfrak{S}_n} \sgn(\sigma) \prod_{i \in [n]} A_{i,\sigma(i)},
\end{align}
where $\mathfrak{S}_n$ denotes the symmetric group on $[n]$, and
$\sgn(\sigma)$ denotes the sign of a permutation $\sigma$.
We define $\det(\mat{A}_{\emptyset}) \triangleq 1$.
For a collection $\mat{V} = \{\vec{v}_1, \ldots, \vec{v}_n\}$ of $n$ vectors in $\bbR^{d}$,
the \emph{volume} of the parallelepiped spanned by $\mat{V}$ is defined as follows:
\begin{align}
\label{eq:def-vol}
    \vol(\mat{V}) \triangleq \|\vec{v}_1\| \cdot
    \prod_{2 \leq i \leq n}
    \dis(\vec{v}_i, \{\vec{v}_1, \ldots, \vec{v}_{i-1}\}).
\end{align}
Here, $\dis(\vec{v}, \mat{P})$ denotes the distance of $\vec{v}$ to the subspace spanned by $\mat{P}$; i.e.,
$
    \dis(\vec{v}, \mat{P}) \triangleq \| \vec{v} - \proj_{\mat{P}}(\vec{v}) \|,
$
where $\proj_{\mat{P}}(\cdot)$ is an operator of orthogonal projection onto the subspace spanned by $\mat{P}$.
We define $\vol(\emptyset) \triangleq 1$
for the sake of consistency to
the determinant of an empty matrix
(i.e., $\det([\;]) = 1 = \vol^2(\emptyset)$).
Note that any symmetric positive semi-definite matrix is a Gram matrix.
Then, if $\mat{A}$ is the Gram matrix defined as
$A_{i,j} \triangleq \langle \vec{v}_i, \vec{v}_j \rangle$
for all $i,j \in [n]$,
we have a simple relation between
the principal minor and the volume of the parallelepiped that
\begin{align}\label{eq:pre:det-vol}
    \det(\mat{A}_S) = \vol^2(\{\vec{v}_i : i \in S\})
\end{align}
for every $S \subseteq [n]$; see \cite{civril2009selecting} for the proof.
We formally define the \prb{Determinant Maximization} problem as follows.\footnote{
Note that if we consider the decision version of \prb{Determinant Maximization},
we are additionally given a target number $\tau$ and
are required to decide if $\maxdet(\mat{A},k) \geq \tau$.
}
\begin{problem}\label{prb:detmax}
Given a positive semi-definite matrix $\mat{A}$ in $\bbQ^{n \times n}$ and a positive integer $k \in [n]$,
\prb{Determinant Maximization} asks to find a set $S \in {[n] \choose k}$ such that
the determinant $\det(\mat{A}_S)$ of a $k \times k$ principal submatrix is maximized.
The optimal value is denoted
$\maxdet(\mat{A},k) \triangleq \max_{S \in {[n] \choose k}} \det(\mat{A}_S)$.
\end{problem}
Due to the equivalence between squared volume and determinant in \cref{eq:pre:det-vol},
\prb{Determinant Maximization} is equivalent to the following problem of volume maximization:
Given a collection of $n$ vectors in $\bbQ^d$ and a positive integer $k \in [n]$,
we are required to find $k$ vectors such that 
the volume of the parallelepiped spanned by them is maximized.
We shall use the problem definition based on the determinant and the volume interchangeably.

\subsection{Parameterized Complexity}
Given a \emph{parameterized problem} $\Pi$ consisting of
a pair $\langle \calI, k \rangle$ of instance $\calI$ and parameter $k \in \bbN$,
we say that $\Pi$ is \emph{fixed-parameter tractable} (FPT)
with respect to $k$ if
it is solvable in $f(k) \abs{\calI}^{\bigO(1)}$ time
for some computable function $f$,
and 
\emph{slice-wise polynomial} (XP) if 
it is solvable in $\abs{\calI}^{f(k)}$ time;
it holds that \FPT~$\subsetneq$~\XP \cite{downey2012parameterized}.
The value of parameter $k$ may be independent of
the instance size $\abs{\calI}$ and may be given by
some computable function $k = k(\calI)$ on instance $\calI$
(e.g., the rank of an input matrix).
Our objective is
to prove that a problem
(i.e., \prb{Determinant Maximization})
is unlikely to admit an FPT algorithm under plausible assumptions in parameterized complexity.
The central notion for this purpose is a parameterized reduction, which is used to demonstrate that a problem of interest is hard
for a particular class of parameterized problems
that is believed to be a superclass of \FPT.
We say that
a parameterized problem $\Pi_1$ is \emph{parameterized reducible}
to another parameterized problem $\Pi_2$ if 
(i)~an instance $\calI_1$ with parameter $k_1$ for $\Pi_1$ can be transformed into
an instance $\calI_2$ with parameter $k_2$ for $\Pi_2$ in FPT time and
(ii)~the value of $k_2$ only depends on the value of $k_1$.
Note that a parameterized reduction may not be a polynomial-time reduction and vice versa.
\W{1} is a class of parameterized problems that
are parameterized reducible to \prb{$k$-Clique}, and
it is known that \FPT~$\subseteq$~\W{1}~$\subseteq$~\XP.
This class is often regarded as a parameterized counterpart to \NP of classical complexity;
in particular, the conjecture \FPT~$\neq$~\W{1} is a widely-believed assumption in parameterized complexity
\cite{flum2006parameterized,downey2012parameterized}.
Thus, the existence of a parameterized reduction from
a \W{1}-complete problem to a problem $\Pi$
is a strong evidence that $\Pi$ is not in \FPT.
In \prb{Determinant Maximization},
a simple brute-force search algorithm that examines
all ${[n] \choose k}$ subsets of size $k$ runs in 
$n^{k+\bigO(1)}$ time; hence,
this problem belongs to \XP.
On the other hand, it is proven to be \W{1}-hard \cite{koutis2006parameterized}.

\section{\W{1}-hardness and \NP-hardness on Arrowhead Matrices}
\label{sec:arrowhead}

We first prove the 
\W{1}-hardness with respect to $k$ and \NP-hardness on arrowhead matrices.
A square matrix $\mat{A}$ in $\bbR^{[0 \isep n] \times [0 \isep n]}$ is
an \emph{arrowhead matrix} if
$A_{i,j} = 0$ for all $i, j \in [n]$ with $i \neq j$.
In the language of graph theory,
$\mat{A}$ is arrowhead if
its symmetrized graph is a star $K_{1,n}$.
See \cref{fig:arrowhead} for the structure of arrowhead matrices.

\begin{theorem}\label{thm:arrowhead}
\prb{Determinant Maximization}  on arrowhead matrices
is \NP-hard and \W{1}-hard when parameterized by $k$.
\end{theorem}
The proof of \cref{thm:arrowhead} requires a reduction
from \prb{$k$-Sum},
a natural parameterized version of the \NP-complete \prb{Subset Sum} problem, 
whose
membership of \W{1} and \W{1}-hardness was proven by
\citet{abboud2014losing} and \citet{downey1995fixed}, respectively.
\begin{problem}[\prb{$k$-Sum} due to \citet{abboud2014losing}]
\label{def:k-sum}
Given $n$ integers $x_1, \ldots, x_n \in [0\isep n^{2k}]$,
a target integer $t \in [0 \isep n^{2k}]$, and 
a positive integer $k \in [n]$,
we are required to decide if
there exists a size-$k$ set $S \in {[n] \choose k}$ such that
$
\sum_{i \in S} x_i = t.
$
\end{problem}
Here, we introduce a slightly-modified version of
\prb{$k$-Sum} such that 
the input numbers are \emph{rational} and
their sum is \emph{normalized} to $1$,
without affecting its computational complexity.

\begin{problem}[\prb{$k$-Sum} modified from \cite{abboud2014losing}]
\label{def:k-sum2}
Given $n$ rational numbers $x_1, \ldots, x_n$ in $(0,1) \cap \bbQ_+$,
a target rational number $t$ in $(0,1) \cap \bbQ_+$, and 
a positive integer $k \in [n]$ such that
$x_i$'s are integer multiples of some rational number at least $\frac{1}{n^{2k+1}}$ and
$\sum_{i \in [n]} x_i = 1$,
\prb{$k$-Sum} asks to decide if
there exists a set $S \in {[n] \choose k}$ such that
$
\sum_{i \in S} x_i = t.
$
\end{problem}
Hereafter, for any set $S \subseteq [0\isep n]$ including $0$,
we denote $S_{-0} \triangleq S \setminus \{0\}$.

\subsection{Reduction from {\normalfont\prb{$k$-Sum}} and Proof of \cref{thm:arrowhead}}
\label{subsec:arrowhead:reduction}

In this subsection, we give
a parameterized, polynomial-time reduction from \prb{$k$-Sum}.
We first use an explicit formula of the determinant of arrowhead matrices.

\begin{lemma}
\label{lem:arrowhead:det}
Let $\mat{A}$ be an arrowhead matrix in $\bbR^{[0\isep n] \times [0 \isep n]}$
such that
$A_{i,i} \neq 0$ for all $i \in [n]$.
Then, for any set $S \subseteq [0 \isep n]$, it holds that
\begin{align}
\det(\mat{A}_S) = 
\begin{dcases}
\prod_{i \in S_{-0}} A_{i,i} \cdot
\left( A_{0,0} - \sum_{i \in S_{-0}} \frac{A_{0,i} \cdot A_{i,0}}{A_{i,i}} \right) 
& \text{if } 0 \in S, \\
\prod_{i \in S} A_{i,i} & \text{if } 0 \not\in S.
\end{dcases}
\end{align}
\end{lemma}
\begin{proof}
The case of $0 \not \in S$ is evident because $\mat{A}_S$ is diagonal.
Showing the case of $S \triangleq [0 \isep n]$ suffices to complete the proof.
Here, we enumerate permutations $\sigma \in \mathfrak{S}_{S} $
such that $A_{i,\sigma(i)}$ is possibly nonzero for all $i \in S$
by the case analysis of $\sigma(0)$.
\begin{description}
\item[\textbf{Case 1.}] If $\sigma(0) = 0$: we must have $\sigma(i) = i$ for all $i \in [n]$ and $\sgn(\sigma) = +1$.
\item[\textbf{Case 2.}] If $\sigma(0) = i$ for $i \neq 0$: we must have $\sigma(i) = 0$, and thus, it holds that $\sigma(j) = j$ for all $j \in S \setminus \{i\}$ and $\sgn(\sigma) = -1$.
\end{description}
Expanding $\det(\mat{A}_S)$, we derive
\begin{align}
\begin{aligned}
    \det(\mat{A}_S) & = 
\sum_{\sigma \in \mathfrak{S}_{S}} \sgn(\sigma) \prod_{i \in S} A_{i,\sigma(i)} \\
& = \prod_{i \in S} A_{i,i} - \sum_{i \in [n]} A_{0,i} \cdot A_{i,0} \prod_{j \in [n] \setminus \{i\}} A_{j,j} \\
& = \prod_{i \in S_{-0}} A_{i,i} \cdot
\left( A_{0,0} - \sum_{i \in S_{-0}} \frac{A_{0,i} \cdot A_{i,0}}{A_{i,i}} \right),
\end{aligned}
\end{align}
completing the proof.
\end{proof}

\cref{lem:arrowhead:det} shows us a way to express
the \emph{product} of
$\exp\left(\sum_{i \in S_{-0}} x_i\right)$ and $1 - C \cdot \sum_{i \in S_{-0}} x_i$ for some constant $C$,
which is a key step in proving \cref{thm:arrowhead}.
Specifically, given $n$ rational numbers $x_1, \ldots, x_n$ and a target rational number $t$ as a \prb{$k$-Sum} instance,
we construct $n+1$ $2n$-dimensional vectors $\vec{v}_0, \ldots, \vec{v}_n $ in $\bbR_+^{2n}$,
each entry of which is defined as follows:
\begin{align}
\label{eq:arrowhead:vector}
v_0(j) =
\begin{cases}
\gamma \cdot \sqrt{x_j} & \text{if } j \leq n, \\
0 & \text{otherwise},
\end{cases}
v_i(j) =
\begin{cases}
\sqrt{\alpha \cdot \rme^{x_i}} & \text{if } j = i, \\
\sqrt{\beta \cdot \rme^{x_i}} & \text{if } j = i+n, \\
0 & \text{otherwise},
\end{cases}
\text{ for all } i \in [n],
\end{align}
where $\alpha$, $\beta$, and $\gamma$ are parameters,
whose values are positive and will be determined later.
We calculate the principal minor of the Gram matrix defined by $\vec{v}_0, \ldots, \vec{v}_n$ as follows.

\begin{lemma}
\label{lem:arrowhead:Gram-det}
Let $\mat{A}$ be the Gram matrix defined by
$n+1$ vectors $\vec{v}_0, \ldots, \vec{v}_n$
that are constructed from an instance of \prb{$k$-Sum} by \cref{eq:arrowhead:vector}.
Then, $\mat{A}$ is an arrowhead matrix, and
for any set $S \subseteq [0 \isep n]$, it holds that
\begin{align}
\det(\mat{A}_S) =
\begin{dcases}
(\alpha + \beta)^{\abs{S}-1} \cdot \gamma^2 \cdot \exp\left(\sum_{i \in S_{-0}}x_i\right) \cdot
\left(1-\frac{\alpha}{\alpha+\beta} \sum_{i \in S_{-0}} x_i\right)&
\text{if } 0 \in S, \\
(\alpha + \beta)^{\abs{S}} \cdot \exp\left(\sum_{i \in S} x_i\right) & \text{if } 0 \not\in S.
\end{dcases}
\end{align}
Moreover, if we regard the principal minor $\det(\mat{A}_S)$ in the case of $0 \in S$
as a function in $X \triangleq \sum_{i \in S_{-0}} x_i$,
it is maximized when $X = \frac{\beta}{\alpha}$.
\end{lemma}
\begin{proof}
Observe first that the inner product between each pair of the vectors (i.e., each entry of $\mat{A}$) is calculated as follows:
\begin{align}
\begin{aligned}
    \langle \vec{v}_0, \vec{v}_0 \rangle & = A_{0,0} = \gamma^2 \sum_{i \in [n]} x_i = \gamma^2, \\
    \langle \vec{v}_0, \vec{v}_i \rangle & = A_{0,i} = A_{i,0} = \gamma \cdot \sqrt{\alpha \cdot x_i \cdot \rme^{x_i}} & \text{ for all } i \in [n], \\
    \langle \vec{v}_i, \vec{v}_i \rangle & = A_{i,i} = (\alpha + \beta) \cdot \rme^{x_i} & \text{ for all } i \in [n], \\
    \langle \vec{v}_i, \vec{v}_j \rangle & = A_{i,j} = 0 & \text{ for all } i \neq j \in [n].
\end{aligned}
\end{align}
Thus, $\mat{A}$ is an arrowhead matrix.
According to \cref{lem:arrowhead:det}, for any set $S \subseteq [0 \isep n]$ including $0$, we have
\begin{align}
\label{eq:arrowhead:Gram-det-0}
\begin{aligned}
    \det(\mat{A}_S)
    & = \prod_{i \in S_{-0}} A_{i,i} \cdot
    \left(A_{0,0} - \sum_{i \in S_{-0}} \frac{A_{0,i} \cdot A_{i,0}}{A_{i,i}} \right)
    \\
    & = \left(\prod_{i \in S_{-0}} (\alpha + \beta) \cdot \rme^{x_i}\right)\cdot 
    \left(\gamma^2 - \sum_{i \in S_{-0}} \frac{\gamma^2 \cdot \alpha \cdot x_i \cdot \rme^{x_i}}{(\alpha + \beta) \cdot \rme^{x_i}} \right) 
    \\
    & = (\alpha + \beta)^{\abs{S}-1} \cdot \gamma^2 \cdot \exp\left(\sum_{i \in S_{-0}} x_i\right) \cdot \left(1-\frac{\alpha}{\alpha + \beta} \sum_{i \in S_{-0}}x_i\right).
\end{aligned}
\end{align}
On the other hand, if $0 \not \in S$, we have
\begin{align}
    \det(\mat{A}_S) = 
    \prod_{i \in S} A_{i,i} = (\alpha + \beta)^{\abs{S}} \cdot \exp\left(\sum_{i \in S} x_i\right).
\end{align}
Setting the derivative of \cref{eq:arrowhead:Gram-det-0} by a variable $X \triangleq \sum_{i \in S} x_i$ equal to $0$,
we obtain
\begin{align}
\begin{aligned}
& \frac{\partial}{\partial X} \left\{ (\alpha + \beta)^{\abs{S}-1} \cdot \gamma^2 \cdot \rme^X \cdot \left(1-\frac{\alpha}{\alpha + \beta} X\right) \right\} = 0 \\
\implies &
\rme^X \left(1-\frac{\alpha}{\alpha + \beta} X\right) + \rme^X \left(-\frac{\alpha}{\alpha + \beta}\right) = \rme^X \cdot \frac{\beta - \alpha X}{\alpha + \beta} = 0 \\
\implies & X = \frac{\beta}{\alpha}.
\end{aligned}
\end{align}
This completes the proof.
\end{proof}

We now determine the values of $\alpha$, $\beta$, and $\gamma$.
Since \cref{lem:arrowhead:Gram-det} demonstrates that
the principal minor for $S$ including $0$ is maximized when
$ \sum_{i \in S_{-0}} x_i = \frac{\beta}{\alpha}$,
we fix $\alpha \triangleq 1$ and $\beta \triangleq t$.
We define $ \delta \triangleq \frac{1}{n^{2k+1}}$,
denoting a lower bound on the \emph{minimum possible absolute difference} between
any sum of $x_i$'s; i.e.,
$\abs{\sum_{i \in S}x_i - \sum_{i \in T}x_i} \geq \delta$ for any $S, T \subseteq [n]$
whenever $\sum_{i \in S}x_i \neq \sum_{i \in T}x_i$.
For the correctness of the value of $\delta$, refer to the definition of \cref{def:k-sum2}.
We finally fix the value of $\gamma$ as $\gamma \triangleq 5$,
so that
\begin{align}
\label{eq:arrowhead:gamma}
    (1+t)^2 \cdot \rme^{1-t} \cdot \frac{1}{\rme^{-\delta} \cdot (1+\delta)}
    \leq 2^2 \cdot \rme \cdot \frac{1}{\rme^{-\frac{1}{2}} \cdot 1}
    < 25 = \gamma^2.
\end{align}
The above inequality ensures that 
$\det(\mat{A}_S)$ is ``sufficiently'' small whenever $0 \not\in S$,
as validated in the following lemma.

\begin{lemma}
\label{lem:arrowhead:OPT}
Let $\mat{A}$ be the Gram matrix defined by $n+1$ vectors
constructed according to \cref{eq:arrowhead:vector},
where $\alpha=1$, $\beta=t$, and $\gamma=5$.
Define $\OPT \triangleq (1+t)^{k-1} \cdot \gamma^2 \cdot \rme^t$.
Then, for any set $S \in {[0\isep n] \choose k+1}$, 
\begin{align}
    \det(\mat{A}_S)
    \text{ is }
    \begin{cases}
        \text{equal to } \OPT & \text{if } 0 \in S \text{ and } \sum_{i \in S_{-0}}x_i = t, \\
        \text{at most } \rme^{-\delta} (1+\delta) \cdot \OPT & \text{otherwise,}
    \end{cases}
\end{align}
where $\delta = \frac{1}{n^{2k+1}}$.
In particular,
$\maxdet(\mat{A},k+1)$ is
$\OPT$ if \prb{$k$-Sum} has a solution, and
is at most $\rme^{-\delta} (1+\delta) \cdot \OPT < \OPT$ otherwise.
\end{lemma}
\begin{proof}
For any set $S \in {[0 \isep n] \choose k+1}$
such that $0 \in S$ and $\sum_{i \in S_{-0}} x_i = t$
(i.e., \prb{$k$-Sum} has a solution),
\cref{lem:arrowhead:Gram-det} derives that
\begin{align}
\det(\mat{A}_S) = (\alpha + \beta)^k \cdot \gamma^2 \cdot \exp(t) \cdot \left(1-\frac{\alpha}{\alpha + \beta} \cdot t\right) = \OPT,
\end{align}
which is the maximum possible principal minor under $0 \in S$.
For any set $S \in {[n] \choose k+1}$ excluding $0$,
by definition of $\gamma$ and \cref{eq:arrowhead:gamma},
we obtain
\begin{align}
\label{eq:arrowhead:OPT:0}
\begin{aligned}
\det(\mat{A}_S) & = (1+t)^{k+1} \cdot \exp\left(\sum_{i \in S_{-0}}x_i\right) \\
    & \leq (1+t)^{k+1} \cdot \rme \cdot
    \left\{ \gamma^2 \cdot \underbrace{\frac{1}{(1+t)^2} \cdot \rme^{t-1} \cdot \rme^{-\delta} \cdot (1+\delta)}_{\geq \gamma^{-2}} \right\}
    \\
    & = \rme^{-\delta}(1+\delta) \cdot \underbrace{(1+t)^{k-1} \cdot \rme^t \cdot \gamma^2}_{= \OPT} \\
    & = \rme^{-\delta} (1+\delta) \cdot \OPT.
\end{aligned}
\end{align}
We now bound $\det(\mat{A}_S)$
for any set $S \in {[0 \isep n] \choose k+1}$ such that
$0 \in S$ and $\sum_{i \in S_{-0}} x_i \neq t$.
Consider first that $\sum_{i \in S_{-0}} x_i$ is greater than $t$;
i.e., $\sum_{i \in S_{-0}} x_i = t + \Delta$ for some $\Delta > 0$.
By \cref{lem:arrowhead:Gram-det}, we have
\begin{align}
\label{eq:arrowhead:OPT:plus}
\begin{aligned}
    \det(\mat{A}_S) & = (1+t)^k \cdot \gamma^2 \cdot \exp(t+\Delta) \cdot \left(1-\frac{t+\Delta}{1+t}\right) \\
    & = \rme^{\Delta}  (1-\Delta) \cdot \OPT \leq \rme^{-\delta}  (1-\delta) \cdot \OPT,
\end{aligned}
\end{align}
where we used the fact that
$\Delta \geq \delta$ by definition of $\delta$ and
$\rme^{\Delta} (1-\Delta)$ is a decreasing function for $\Delta > 0$.
Consider then that
$\sum_{i \in S_{-0}} x_i = t - \Delta$ for some $\Delta > 0$,
which yields that
\begin{align}
\label{eq:arrowhead:OPT:minus}
\begin{aligned}
\det(\mat{A}_S) & = (1+t)^k \cdot \gamma^2 \cdot \exp(t-\Delta) \cdot \left(1-\frac{t-\Delta}{1+t}\right) \\
    & = \rme^{-\Delta}  (1+\Delta) \cdot \OPT \leq \rme^{-\delta} (1+\delta) \cdot \OPT,
\end{aligned}
\end{align}
where we used the fact that $\Delta \geq \delta$ and
$\rme^{-\Delta} (1+\Delta)$ is a decreasing function for $\Delta > 0$.
By combining 
\cref{eq:arrowhead:OPT:0,eq:arrowhead:OPT:plus,eq:arrowhead:OPT:minus},
if $S\in {[0 \isep n] \choose k+1} $ satisfies that 
$0 \not\in S$ or $\sum_{i \in S_{-0}} x_i \neq t$,
its principal minor is bounded as follows:
\begin{align}
    \det(\mat{A}_S)
    \leq \max\Bigl\{ \rme^{-\delta}  (1+\delta), \rme^{\delta}  (1-\delta) \Bigr\} \cdot \OPT
    = \rme^{-\delta}(1+\delta) \cdot \OPT.
\end{align}
Observing that $\rme^{-\delta}(1+\delta) < 1$ for any $\delta > 0$ accomplishes the proof.
\end{proof}

We complete our reduction by
approximating the Gram matrix $\mat{A}$ of $n+1$ vectors defined in \cref{eq:arrowhead:vector}
by a \emph{rational} matrix $\mat{B}$
whose maximum determinant maintains sufficient information to solve  \prb{$k$-Sum}.

\begin{lemma}
\label{lem:arrowhead:approx}
Let $\mat{B}$ be the Gram matrix in $\bbQ^{(n+1) \times (n+1)}$ defined by
$n+1$ vectors $\vec{w}_0, \ldots, \vec{w}_n$ in $\bbQ^{2n}$,
each entry of which is a $(1 \pm \epsilon)$-factor approximation to
the corresponding entry of $n+1$ vectors $\vec{v}_0, \ldots, \vec{v}_n$ defined by \cref{eq:arrowhead:vector},
where $\epsilon = 2^{-\bigO(k \log (nk))}$.
Then,
\begin{align}
\maxdet(\mat{B}, k+1)
\text{ is }
\begin{dcases}
\text{at least } \left(\frac{2}{3} + \frac{1}{3} \rme^{-\delta} (1+\delta)\right) \cdot \OPT & \text{if \prb{$k$-Sum} has a solution,} \\
\text{at most } \left(\frac{1}{3} + \frac{2}{3} \rme^{-\delta} (1+\delta)\right) \cdot \OPT & \text{otherwise.}
\end{dcases}
\end{align}
Moreover, we can calculate $\mat{B}$ in polynomial time.
\end{lemma}
The crux of its proof is to approximate $\mat{A}$ within a factor of $\epsilon = 2^{-\bigO(k \log (nk))}$.
To this end, we use the following lemma.

\begin{lemma}[cf.~\protect{\cite[page~107]{bhatia2007perturbation}}]
\label{lem:perturbation}
For two complex-valued $n \times n$ matrices $\mat{A}$ and $\mat{B}$,
the absolute difference in the determinant of $\mat{A}$ and $\mat{B}$ is bounded from above by
\begin{align}
    \abs{\det(\mat{A}) - \det(\mat{B})} \leq
    n \cdot \max\{\|\mat{A}\|_2, \|\mat{B}\|_2\}^{n-1} \cdot \|\mat{A} - \mat{B}\|_2.
\end{align}
\end{lemma}

\begin{proof}[Proof of \cref{lem:arrowhead:approx}]
Let $n$ rational numbers $x_1, \ldots, x_n$,
a target rational number $t$, and
a positive integer $k$
be an instance of \prb{$k$-Sum}.
Suppose we are given the Gram matrix $\mat{A}$ defined by $n+1$ vectors $\vec{v}_0, \ldots, \vec{v}_n$
constructed according to \cref{eq:arrowhead:vector} and
the rational Gram matrix $\mat{B}$ defined by $n+1$ rational vectors $\vec{w}_0, \ldots, \vec{w}_n$,
each entry of which is a $(1\pm\epsilon)$-factor approximation to the corresponding entry of $\vec{v}_i$'s.
If the absolute difference between
$\mat{A}_S$ and $\mat{B}_S$ is at most
$ \frac{1}{3}( \OPT - \rme^{-\delta}(1+\delta) \cdot \OPT) $ for
every $S \in {[0 \isep n] \choose k+1}$,
we can use \cref{lem:arrowhead:OPT} to ensure that
\begin{align}
    \det(\mat{B}_S) & \geq \left(\frac{2}{3} + \frac{1}{3} \rme^{-\delta}(1+\delta)\right) \cdot \OPT
    \; \text{ if } 0 \in S \text{ and } \sum_{i \in S_{-0}}x_i = t,
    \label{eq:arrowhead:approx:det-inequ-0}
    \\
    \det(\mat{B}_S) & \leq \left(\frac{1}{3} + \frac{2}{3} \rme^{-\delta}(1+\delta)\right) \cdot \OPT
    \; \text{ otherwise}.
\label{eq:arrowhead:approx:det-inequ-1}
\end{align}
In particular, we can use either
the optimal value or solution for \prb{Determinant Maximization} 
defined by $(\mat{B},k+1)$
to determine whether \prb{$k$-Sum} has a solution.

We demonstrate that this is the case if
$\epsilon = 2^{-\bigO(k \log(nk))}$.
Owing to the nonnegativity of $\vec{v}_i$'s and $\vec{w}_i$'s,
we have
$(1-\epsilon) v_i(e) \leq w_i(e) \leq (1+\epsilon) v_i(e)$ for every $i \in [0 \isep n]$ and $e \in [2n]$,
implying that:
\begin{align}
\begin{aligned}
& (1-\epsilon)^2 \cdot v_i(e) v_j(e) \leq w_i(e)w_j(e) \leq (1+\epsilon)^2 \cdot v_i(e) v_j(e) \\
& \implies (1-\epsilon)^2 \cdot
\underbrace{\langle \vec{v}_i, \vec{v}_j \rangle}_{=A_{i,j}} \leq
\underbrace{\langle \vec{w}_i, \vec{w}_j \rangle}_{=B_{i,j}} \leq
(1+\epsilon)^2 \cdot \underbrace{\langle \vec{v}_i, \vec{v}_j \rangle}_{=A_{i,j}}.
\end{aligned}
\end{align}
Because it holds that 
$(1+\epsilon)^2 \leq 1+3\epsilon$ and 
$(1-\epsilon)^2 \geq 1-3\epsilon$ for any $\epsilon \in (0,\frac{1}{3})$,
there exists a number $\rho_{i,j} \in [1-3\epsilon, 1+3\epsilon]$
such that $B_{i,j} = \rho_{i,j} \cdot A_{i,j}$ for each $i,j \in [0 \isep n]$.
By applying \cref{lem:perturbation},
we can bound the absolute difference between
the determinant of $\mat{A}_S$ and $\mat{B}_S$
for any set $S \in {[0\isep n] \choose k+1}$ as:
\begin{align}
    \abs{\det(\mat{A}_S) - \det(\mat{B}_S)}
    \leq (k+1) \cdot \max\{ \|\mat{A}_S\|_2, \|\mat{B}_S\|_2 \}^k \cdot
    \|\mat{A}_S - \mat{B}_S\|_2.
\end{align}
Each term in the above inequality can be bounded as follows:
\begin{align}
    \|\mat{A}_S\|_2 & \leq (k+1) \cdot \|\mat{A}_S\|_{\max}
    \leq (k+1) \cdot \|\mat{A}\|_{\max}, \\
    \begin{split}
    \|\mat{B}_S\|_2 & \leq (k+1) \cdot \|\mat{B}_S\|_{\max} \\
    & \leq (k+1) \cdot \|\mat{A}\|_{\max} \cdot \max_{i,j} \abs{\rho_{i,j}} \\
    & \leq (k+1) (1+3\epsilon) \cdot \|\mat{A}\|_{\max},
    \end{split} \\
    \begin{split}
    \|\mat{A}_S - \mat{B}_S\|_2 & \leq (k+1) \cdot \|\mat{A}_S - \mat{B}_S\|_{\max} \\
    & \leq (k+1) \cdot \max_{i,j \in S} \abs{A_{i,j} - \rho_{i,j}A_{i,j}} \\
    & \leq (k+1) 3\epsilon \cdot \|\mat{A}\|_{\max}.
    \end{split}
\end{align}
Here, $\|\mat{A}\|_{\max}$ is bounded using its definition
(see the beginning of the proof of \cref{lem:arrowhead:Gram-det}):
\begin{align}
\begin{aligned}
    \|\mat{A}\|_{\max} & = \max\left\{\max_{i \in [n]} \; (1+t) \cdot \rme^{x_i}, \gamma^2, \max_{i \in [n]} \; \gamma \sqrt{\alpha \cdot x_i \cdot \rme^{x_i}} \right\} \\
    & \leq \max\Bigl\{ 2 \rme, 25, 5 \sqrt{\rme} \Bigr\} = 25.
\end{aligned}
\end{align}
Putting it all together, we get
\begin{align}
\begin{aligned}
    \abs{\det(\mat{A}_S) - \det(\mat{B}_S)}
    & \leq (k+1) \cdot \{ (k+1) \cdot 25 \cdot (1+3\epsilon) \}^k \cdot (k+1) \cdot 25 \cdot 3\epsilon \\
    & = (k+1)^{k+2} \cdot 25^{k+1} \cdot (1+3\epsilon)^k \cdot 3\epsilon.
\end{aligned}
\end{align}
Therefore, for the absolute difference of the determinant
between $\mat{A}_S$ and $\mat{B}_S$ to be less than
$\frac{1}{3}(1 - \rme^{-\delta}(1+\delta)) \cdot \OPT$,
the value of $\epsilon$ should be less than
\begin{align}
\label{eq:arrowhead:approx:eps-inequ}
    \epsilon < \Bigl\{ (k+1)^{k+2} \cdot 25^{k+1} \cdot (1+3\epsilon)^k \cdot 3 \Bigr\}^{-1} \cdot \frac{1}{3} (1-\rme^{-\delta} (1+\delta))\cdot \OPT.
\end{align}
Observe that each term in \cref{eq:arrowhead:approx:eps-inequ} can be bounded as follows:
\begin{align}
    \Bigl\{ (k+1)^{k+2} \cdot 25^{k+1} \cdot (1+3\epsilon)^k \cdot 3 \Bigr\}^{-1} & > (100000k)^{-k} \;\text{ for all } k \geq 1 \text{ if } \epsilon \leq 1, \\
    1-\rme^{-\delta}(1+\delta) & \geq \frac{\delta^2}{2} - \frac{\delta^3}{3} > 0 \;\text{ for any } \delta > 0, \\
    \OPT & \geq 1.
\end{align}
Consequently, we can set the value of $\epsilon$ so as to satisfy
\cref{eq:arrowhead:approx:eps-inequ}; thus, \cref{eq:arrowhead:approx:det-inequ-0,eq:arrowhead:approx:det-inequ-1}:
\begin{align}
    \epsilon \triangleq (100000k)^{-k} \cdot \left(\frac{\delta^2}{2} - \frac{\delta^3}{3} \right)
    = 2^{-\bigO(k \log (nk))}.
\end{align}
We finally claim that each entry of $\vec{w}_i$'s can be computed in polynomial time.
Because of the definition of $\vec{v}_i$'s in \cref{eq:arrowhead:vector},
it suffices to compute a $(1\pm \frac{\epsilon}{2})$-approximate value of $\exp(x)$ and $\sqrt{x}$ for a rational number $x$
in polynomial time in the input size and
$\log \epsilon^{-1} = \bigO(k \log (nk))$, completing the proof.\footnote{
Indeed, we can simply take the sum of the first $\bigO(\log \epsilon^{-1} + x)$ terms of a Taylor series of $\exp(x)$
to calculate its $(1 \pm \epsilon)$-approximation;
we can use a bisection search, 
of which number of iterations is at most
$\bigO(\log \epsilon^{-1} + \log x^{-1})$,
to approximate $\sqrt{x}$ within a $(1 \pm \epsilon)$-factor.
}
\end{proof}

What remains to be done is to prove \cref{thm:arrowhead} using \cref{lem:arrowhead:approx}.
\begin{proof}[Proof of \cref{thm:arrowhead}]
Our parameterized reduction is as follows.
Given $n$ rational numbers $x_1, \ldots, x_n \in (0,1) \cap \bbQ$, a target rational number $t \in (0,1) \cap \bbQ$, and a positive integer $k \in [n]$
as an instance of \prb{$k$-Sum},
we construct
$n+1$ rational vectors $\vec{w}_0, \ldots, \vec{w}_n$ in $\bbQ_+^{2n}$,
each of which is an entry-wise $(1 \pm \epsilon)$-factor approximation
to $\vec{v}_0, \ldots, \vec{v}_n$ defined by \cref{eq:arrowhead:vector}, where $\epsilon = 2^{-\bigO(k \log (nk))}$.
This construction requires polynomial time owing to \cref{lem:arrowhead:approx}.
Thereafter, we compute the Gram matrix $\mat{B}$ in $\bbQ^{(n+1) \times (n+1)}$ defined by $\vec{w}_0, \ldots, \vec{w}_n$.
Consider \prb{Determinant Maximization} defined by $(\mat{B},k+1)$
with parameter $k+1$.
According to \cref{lem:arrowhead:approx},
the maximum principal minor
$\maxdet(\mat{B},k+1)$
is at least
$(\frac{2}{3} + \frac{1}{3}\rme^{-\delta}(1+\delta)) \cdot \OPT$ if and only if \prb{$k$-Sum} has a solution.
Moreover, if this is the case,
the optimal solution $S^*$ for \prb{Determinant Maximization}
satisfies that $\sum_{i \in S^*_{-0}} x_i = t$.
The above discussion ensures the correctness of the parameterized reduction from \prb{$k$-Sum} to \prb{Determinant Maximization}, finishing the proof.
\end{proof}

\subsection{Note on Polynomial-time Solvability for Tridiagonal Matrices and Spiders of Bounded Legs \cite{al-thani2021tridiagonal,al-thani2021tridiagonala}}
\label{subsec:tridiagonal}
Here, we mention some tractable cases of \prb{Determinant Maximization} due to \citet{al-thani2021tridiagonal,al-thani2021tridiagonala}.
Recall that a \emph{tridiagonal matrix} is a square matrix $\mat{A}$ such that $A_{i,j} = 0$ whenever $\abs{i-j} \geq 2$;
i.e.,
its symmetrized graph is a path graph (and thus a linear forest).
A graph is called a \emph{spider} if it is a tree having 
at most one vertex of degree greater than $2$,
and its leafs are called \emph{legs}.

\begin{observation}[\protect{\citet{al-thani2021tridiagonal,al-thani2021tridiagonala}}]
\label{obs:tridiagonal}
\prb{Determinant Maximization} can be solved in polynomial time
if an input matrix is a tridiagonal matrix, or
its symmetrized graph is a spider with a constant number of legs.
\end{observation}

\section{\W{1}-hardness With Respect to Rank}
\label{sec:rank}

We then prove the \W{1}-hardness of
\prb{Determinant Maximization} when parameterized by the \emph{rank} of an input matrix.
In fact, we obtain the stronger hardness result
on the problem of finding a set of pairwise orthogonal rational vectors,
which is formally stated below.

\begin{problem}\label{prb:orthogonal}
Given $n$ $d$-dimensional vectors $\vec{v}_1, \ldots, \vec{v}_n$ in $\bbQ^d$ and a positive integer $k \in [n]$,
we are required to decide if there exists a set of $k$ vectors that is pairwise orthogonal, i.e.,
a set $S \in {[n] \choose k}$ such that
$\langle \vec{v}_i, \vec{v}_j \rangle = 0$ for all $i \neq j \in S$.
\end{problem}

\begin{theorem}\label{thm:rank}
\cref{prb:orthogonal} is \W{1}-hard when parameterized by 
the dimension $d$ of the input vectors.
Moreover, the same hardness result holds even if 
every vector has the same Euclidean norm.
\end{theorem}\noindent
The following is immediate from \cref{thm:rank}.
\begin{corollary}
\label{cor:rank}
\prb{Determinant Maximization} is \W{1}-hard when parameterized by the rank of an input matrix.
\end{corollary}
\begin{proof}
Let $\mat{A}$ be the Gram matrix defined by
any $n$ $d$-dimensional vectors $\vec{v}_1, \ldots, \vec{v}_n \in \bbQ^d$ having the same Euclidean norm, say, $c \in \bbQ_+$.
Consider \prb{Determinant Maximization} defined by $(\mat{A}, k)$.
For any $S \in {[n] \choose k}$,
the principal minor $\det(\mat{A}_S)$ is equal to $c^{2k}$ if
the set of $k$ vectors $\{\vec{v}_i : i \in S\}$ is pairwise orthogonal and
is strictly less than $c^{2k}$ otherwise.
Observing that $\rank(\mat{A}) \leq d$ completes the proof.
\end{proof}

Unlike the proof of \cref{thm:arrowhead},
$ f(k) $-dimensional vectors can only be used in a parameterized reduction.
The key tool to bypass this difficulty
is \prb{Grid Tiling} introduced in the next subsection.

\subsection{{\normalfont \prb{Grid Tiling}} and Pythagorean Triples}
We first define \prb{Grid Tiling} due to 
\citet{marx2007optimality}.

\begin{table}
    \centering
    \begin{tabular}{|c|c|c|}
        \hline
        $S_{1,1}$ & $S_{2,1}$ & $S_{3,1}$ \\
        $(1,1)$ & $\color{magenta}{\mathbf{(1,2)}}$ & $(1,2)$ \\
        $\color{magenta}{\mathbf{(3,2)}}$ & $(2,2)$ & $\color{magenta}{\mathbf{(4,2)}}$ \\
        \hline
        $S_{1,2}$ & $S_{2,2}$ & $S_{3,2}$ \\
        $\color{magenta}{\mathbf{(3,4)}}$ & $\color{magenta}{\mathbf{(1,4)}}$ & $(2,1)$ \\
        $(4,3)$ & $(3,1)$ & $\color{magenta}{\mathbf{(4,4)}}$ \\
        \hline
        $S_{1,3}$ & $S_{2,3}$ & $S_{3,3}$ \\
        $\color{magenta}{\mathbf{(3,2)}}$ & $\color{magenta}{\mathbf{(1,2)}}$ & $(3,3)$ \\
        $(4,1)$ & $(1,3)$ & $\color{magenta}{\mathbf{(4,2)}}$ \\
        \hline
    \end{tabular}
    \caption{
    Example of \prb{Grid Tiling} with $k=3$ and $n=4$.
    The colored pairs form a solution; i.e.,
    $\sigma(1,1)=(3,2)$,
    $\sigma(2,1)=(1,2)$,
    $\sigma(3,1)=(4,2)$,
    $\sigma(1,2)=(3,4)$,
    $\sigma(2,2)=(1,4)$,
    $\sigma(3,2)=(4,4)$,
    $\sigma(1,3)=(3,2)$,
    $\sigma(2,3)=(1,2)$, and
    $\sigma(3,3)=(4,2)$.
    Observe that
    the first coordinate of the selected pairs in the first column is $3$,
    the second coordinate of the selected pairs in the first row is $2$, and so on.
    }
    \label{tab:grid-tiling}
\end{table}

\begin{problem}[\prb{Grid Tiling} due to 
\citet{marx2007optimality}]
\label{def:grid-tiling}
For two integers $n$ and $k$,
given a collection $\calS$ of $k^2$ nonempty sets $S_{i,j} \subseteq [n]^2$ for each $i,j \in [k]$,
\prb{Grid Tiling} asks to find
an assignment $\sigma \colon [k]^2 \to [n]^2$
with $\sigma(i,j) \in S_{i,j}$ such that
\begin{enumerate}
\item vertical neighbors agree in the first coordinate; i.e.,
if $\sigma(i,j) = (x,y)$ and
$\sigma(i,(j+1) \bmod k) = (x',y')$, then $x = x'$, and
\item horizontal neighbors agree in the second coordinate; i.e.,
if $\sigma(i,j) = (x,y)$ and $\sigma((i+1) \bmod k,j) = (x',y')$, then $y = y'$,
\end{enumerate}
where we define
$(k+1) \bmod k \triangleq 1$, and hereafter omit the symbol $\mathrm{mod}$ for modulo operator.
Each pair $(i,j) \in [k]^2$ will be referred to as a \emph{cell}.
\end{problem}

See also \cref{tab:grid-tiling} for an example.
\prb{Grid Tiling} parameterized by $k$ is proven to be \W{1}-hard by \citet{marx2007optimality,marx2012tight}.
We say that two cells $(i_1,j_1)$ and $(i_2,j_2)$ are \emph{adjacent} if
the Manhattan distance between them is $1$.
Let $\calI$ be the set of all pairs of two adjacent cells; i.e.,
\begin{align}
\label{eq:rank:calI}
\calI \triangleq \Bigl\{ (i_1,j_1,i_2,j_2) \in [k]^4 : \abs{i_1-i_2} + \abs{j_1-j_2} = 1
\Bigr\}.
\end{align}
Note that $\abs{\calI} = 2k^2$.
\prb{Grid Tiling} has the two useful properties that 
(i) the constraint to be satisfied is the equality on the first and second coordinates, which is pretty simple, and
(ii) there are only $k^2$ cells and
each cell is adjacent to (at most) \emph{four} cells.
To represent the \emph{consistency}
between adjacent cells using only $f(k)$-dimensional vectors,
we use a rational point
$(\frac{a}{c}, \frac{b}{c})$ on the unit circle
generated from a Pythagorean triple $(a,b,c)$.
A \emph{Pythagorean triple} is
a triple of three positive integers $(a,b,c)$ such that
$a^2 + b^2 = c^2$; e.g., $(a,b,c) = (3,4,5)$.
It is further said to be \emph{primitive} if $(a,b,c)$ are coprime; i.e.,
$\gcd(a,b) = \gcd(b,c) = \gcd(c,a) = 1$.
We assume for a while that we have $n$ primitive Pythagorean triples, denoted $(a_1, b_1, c_1), \ldots, (a_n, b_n, c_n)$.

\subsection{Reduction from {\normalfont \prb{Grid Tiling}} and Proof of \cref{thm:rank}}
We are now ready to describe a parameterized reduction from
\prb{Grid Tiling} to \cref{prb:orthogonal}.
Given an instance $\calS = (S_{i,j})_{i,j \in [k]}$ of \prb{Grid Tiling},
we define a rational vector for each $(x,y) \in S_{i,j}$,
whose dimension is bounded by some function in $k$.
Each vector consists
of $\abs{\calI} = 2k^2$ blocks (indexed by an element of $\calI$),
each of which is two dimensional and
is either a rational point on the unit circle or the origin $\mathrm{O}$.
Hence, each vector is of dimension $2\abs{\calI} = 4k^2$.
Let $\vec{v}^{(i,j)}_{x,y}$ denote the vector for
an element $(x,y) \in S_{i,j}$ of cell $(i,j) \in [k]^2$,
let $\vec{v}^{(i,j)}_{x,y}(i_1,j_1,i_2,j_2)$ denote
the block of $\vec{v}^{(i,j)}_{x,y}$ corresponding to
each pair of adjacent cells $(i_1,j_1,i_2,j_2) \in \calI$.
Each block is defined as follows:

\begin{align}
\label{eq:rank:def}
\vec{v}^{(i,j)}_{x,y}(e) \triangleq
\begin{dcases}
    \Large\left[-\frac{b_x}{c_x}, \frac{a_x}{c_x}\right] & \text{if } e = (i,j-1,i,j), \\
    \Large\left[\frac{a_x}{c_x}, \frac{b_x}{c_x}\right] & \text{if } e = (i,j,i,j+1), \\
    \Large\left[-\frac{b_y}{c_y}, \frac{a_y}{c_y}\right] & \text{if } e = (i-1,j,i,j), \\
    \Large\left[\frac{a_y}{c_y}, \frac{b_y}{c_y}\right] & \text{if } e = (i,j,i+1,j), \\
    [0,0] & \text{otherwise}.
\end{dcases}
\end{align}
Because each vector contains exactly four points on the unit circle,
its squared norm is equal to $4$.
We denote by $\mat{V}^{(i,j)}$
the set of vectors corresponding to the elements of $S_{i,j}$; i.e.,
$
    \mat{V}^{(i,j)} \triangleq \{ \vec{v}^{(i,j)}_{x,y} : (x,y) \in S_{i,j} \}.
$
We now define an instance $(\mat{V}, K)$ of \cref{prb:orthogonal} as
$
    \mat{V} \triangleq \bigcup_{i,j \in [k]} \mat{V}^{(i,j)}
    \text{ and }
    K \triangleq k^2.
$
Note that $\mat{V}$ consists of  $N \triangleq \sum_{i,j \in [k]}\abs{S_{i,j}}$ vectors.
We prove that
the existence of a set of pairwise orthogonal $k^2$ vectors yields
the answer of \prb{Grid Tiling}.
The key property of the above construction is that 
$\left[-\frac{b_x}{c_x}, \frac{a_x}{c_x}\right]$ and $\left[\frac{a_{x'}}{c_{x'}}, \frac{b_{x'}}{c_{x'}}\right]$
are orthogonal if and only if $x = x'$.

\begin{lemma}
\label{lem:rank:orthogonal}
Let $\mat{V}$ be the set of vectors constructed from an instance
$\calS = (S_{i,j})_{i,j \in [k]}$ of
\prb{Grid Tiling} according to \cref{eq:rank:def}.
Then, \prb{Grid Tiling} has a solution if and only if
\cref{prb:orthogonal} has a solution.
\end{lemma}
\begin{proof}
We first prove the only-if direction.
Suppose the \prb{Grid Tiling} instance $\calS$ has a solution denoted
$\sigma \colon [k]^2 \to [n]^2$.
We show that the set
$\mat{S} \triangleq \{\vec{v}^{(i,j)}_{x,y} : i,j \in [k], (x,y) = \sigma(i,j) \}$ of $k^2$ vectors is pairwise orthogonal.
Observe easily that any two vectors corresponding to \emph{nonadjacent} cells are orthogonal.
We then verify the orthogonality of two vectors corresponding to vertically adjacent cells $(i,j)$ and $(i,j+1)$ for any $i,j \in [k]$.
Calculating the inner product between $\vec{v}^{(i,j)}_{x,y}$ and $\vec{v}^{(i,j+1)}_{x,y'}$ in $\mat{S}$,
where $(x,y) = \sigma(i,j)$ and $(x,y') = \sigma(i,j+1)$ for some $x,y,y' \in [n]$ by assumption,
we obtain that
\begin{align}
    \left\langle \vec{v}^{(i,j)}_{x,y}, \vec{v}^{(i,j+1)}_{x,y'} \right\rangle
    = \left\langle \left[\frac{a_x}{c_x}, \frac{b_x}{c_x}\right], \left[-\frac{b_x}{c_x}, \frac{a_x}{c_x}\right] \right\rangle
    = \frac{a_x b_x - a_x b_x}{c_x^2} = 0.
\end{align}
Similarly,
for two horizontally adjacent cells $(i,j)$ and $(i+1,j)$,
we derive that
$ \langle \vec{v}^{(i,j)}_{x,y}, \vec{v}^{(i+1,j)}_{x',y} \rangle = 0 $,
where $(x,y) = \sigma(i,j)$ and $(x',y) = \sigma(i+1,j)$
for some $x,x',y \in [n]$ by assumption.
This accomplishes the proof for the only-if direction.

We then prove the if direction.
Suppose $\mat{S}$ is a set of $k^2$ vectors from $\mat{V}$ that is pairwise orthogonal.
Observe first that
$\mat{S}$ must include exactly one vector from each $\mat{V}^{(i,j)}$,
because otherwise it includes a pair of vectors
$\vec{v}^{(i,j)}_{x,y}$ and $\vec{v}^{(i,j)}_{x',y'}$ for
some distinct $(x,y) \neq (x',y') \in S_{i,j}$, which is nonorthogonal.
Indeed, their inner product is
\begin{align}
    \left\langle \vec{v}^{(i,j)}_{x,y}, \vec{v}^{(i,j)}_{x',y'} \right\rangle
    = 2 \cdot \frac{a_xa_{x'}+b_xb_{x'}}{c_xc_{x'}}
    + 2 \cdot \frac{a_ya_{y'}+b_yb_{y'}}{c_yc_{y'}}
    > 0.
\end{align}
We can thus define the unique assignment $\sigma(i,j) \triangleq (x,y) \in S_{i,j}$ such that
$\vec{v}^{(i,j)}_{x,y} \in \vec{S}$ for each cell $(i,j) \in [k]^2$.
We show that $\sigma$ is a solution of \prb{Grid Tiling}.
Calculating the inner product between $\vec{v}^{(i,j)}_{x,y}$ and $\vec{v}^{(i,j+1)}_{x',y'}$
for two vertically adjacent cells $(i,j)$ and $(i,j+1)$,
where $(x,y) = \sigma(i,j)$ and $(x',y') = \sigma(i,j+1)$,
we have
\begin{align}
    \left\langle \vec{v}^{(i,j)}_{x,y}, \vec{v}^{(i,j+1)}_{x',y'} \right\rangle
    = \frac{a_{x'} b_x - a_x b_{x'}}{c_x c_{x'}} = 0,
\end{align}
i.e., it must hold that $a_{x'} b_x = a_x b_{x'}$ as $c_x c_{x'} > 0$.
Since $a_x$ and $b_x$ are coprime, $a_{x'}$ must divide $a_x$ and $b_{x'}$ must divide $b_x$;
since $a_{x'}$ and $b_{x'}$ are coprime, $a_x$ must divide $a_{x'}$ and $b_x$ must divide $b_{x'}$,
implying that $a_x = a_{x'}$ and $b_x = b_{x'}$.
Consequently, $x = x'$; i.e., the vertical neighbors agree in the first coordinate.
Similarly, for two horizontally adjacent cells $(i,j)$ and $(i+1,j)$,
we can show that
$a_y = a_{y'}$ and $b_y = b_{y'}$,
where $(x,y) = \sigma(i,j)$ and $(x',y') = \sigma(i+1,j)$,
and thus $y=y'$; i.e.,
the horizontal neighbors agree in the second coordinate.
This accomplishes the proof for the if direction.
\end{proof}

\begin{proof}[Proof of \cref{thm:rank}]
Our parameterized reduction is as follows.
Given an instance $\calS = (S_{i,j})_{i,j \in [k]}$ of \prb{Grid Tiling},
we first generate
$n$ primitive Pythagorean triples $(a_1, b_1, c_1), \allowbreak \ldots, \allowbreak (a_n, b_n, c_n)$.
This can be done efficiently by simply letting
$(a_x,b_x,c_x) \triangleq (2x+1, 2x^2+2x, 2x^2+2x+1)$ for all $x \in [n]$.
We then construct a set $\mat{V}$ of $N$ $4k^2$-dimensional rational vectors from $\calS$ according to \cref{eq:rank:def} in polynomial time,
where $N \triangleq \sum_{i,j \in [k]}\abs{S_{i,j}}$.
According to \cref{lem:rank:orthogonal}, $\calS$ has a solution of \prb{Grid Tiling}
if and only if there exists a set of $k^2$ pairwise orthogonal vectors in $\mat{V}$.
Since \prb{Grid Tiling} is \W{1}-hard with respect to $k$,
\cref{prb:orthogonal} is also \W{1}-hard when parameterized by dimension $d (= 4k^2)$.
Note that every vector is of squared norm $4$,
completing the proof.
\end{proof}

\subsection{\cref{prb:orthogonal} on Nonnegative Vectors is FPT}\label{subsec:rank:cp-rank}
We note that \cref{prb:orthogonal}
is FPT with respect to the dimension if the input vectors are \emph{nonnegative}.
Briefly speaking, \cref{prb:orthogonal} on nonnegative vectors is equivalent to
\prb{Set Packing} parameterized by the size of the universe,
which is easily shown to be FPT.

\begin{observation}
\label{obs:cp-rank}
\cref{prb:orthogonal} is FPT with respect to the dimension
if every input vector is entry-wise nonnegative.
\end{observation}
\begin{proof}
Let $\vec{v}_1, \ldots, \vec{v}_n$ be
$n$ $d$-dimensional nonnegative vectors in $\bbQ_+^d$ and
$k \in [n]$ a positive integer.
Without loss of generality, 
we can assume that $k \leq d$ because otherwise,
there is always no solution.
For each vector $\vec{v}_i$,
we denote by $\nz(\vec{v}_i)$
the set of coordinates of positive entries; i.e.,
$\nz(\vec{v}_i) \triangleq \{e \in [d] : v_i(e) > 0\}$.
Then,
the vector set $\{\vec{v}_i : i \in S\}$ for any $S \subseteq [n]$
is pairwise orthogonal if and only if
$\nz(\vec{v}_i) \cap \nz(\vec{v}_j) = \emptyset$
for every $i \neq j \in S$; i.e.,
the problem of interest is \prb{Set Packing} in which
we want to find $k$ pairwise disjoint sets from the family $ \calF \triangleq \{ \nz(\vec{v}_i) : i \in [n] \} $.
Observing that $\abs{\calF} \leq 2^d$ because duplicates (i.e., $\nz(\vec{v}_i) = \nz(\vec{v}_j)$ for some $i \neq j$) are discarded,
there are at most ${\abs{\calF} \choose k} \leq 2^{dk}$
possible subsets of size $k$.
Hence, we construct $\calF$ in $\bigO(nd)$ time and 
perform an exhaustive search in time $2^{dk} \cdot d^{\bigO(1)} \leq 2^{d^2} \cdot d^{\bigO(1)}$,
completing the proof.
\end{proof}

\section{\W{1}-hardness of Approximation}\label{sec:inapprox}
Our final result is FPT-inapproximability of \prb{Determinant Maximization} as stated below.

\begin{theorem}\label{thm:inapprox}
Under the Parameterized Inapproximability Hypothesis,
it is \W{1}-hard to
approximate \prb{Determinant Maximization}
within a factor of $2^{-c\sqrt{k}}$ for some universal constant $c > 0$
when parameterized by the number $k$ of vectors to be selected.
Moreover, the same hardness result holds even if
the diagonal entries of an input matrix are restricted to $1$.
\end{theorem}
Since the above result relies on the Parameterized Inapproximability Hypothesis,
\cref{subsec:inapprox:PIH-MGT} begins with its formal definition.

\subsection{Inapproximability of {\normalfont \prb{Grid Tiling}} under Parameterized Inapproximability Hypothesis}\label{subsec:inapprox:PIH-MGT}
We first introduce \prb{Binary Constraint Satisfaction Problem},
for which the Parameterized Inapproximability Hypothesis
asserts FPT-inapproximability.
For two integers $n$ and $k$,
we are given
a set $V \triangleq [k]$ of $k$ variables,
an alphabet $\Sigma \triangleq [n]$ of size $n$, and 
a set of constraints $\calC = (C_{i,j})_{i,j \in V}$ such that $C_{i,j} \subseteq \Sigma^2$.\footnote{
Though each constraint is actually indexed by
an unordered pair of variables $\{i,j\}$,
we use the present notation $C_{i,j}$ for sake of clarity
and assume that $C_{i,j} = C_{j,i}$ without loss of generality.}
Each variable $i \in V$ may take a value from $\Sigma$.
Each constraint $C_{i,j}$ specifies the pairs of values
that variables $i$ and $j$ can take simultaneously, and
it is said to be \emph{satisfied} by 
an assignment $\psi \colon V \to \Sigma$ of values to the variables
if $(\psi(i), \psi(j)) \in C_{i,j}$.
For example,
for a graph $G = (V,E)$,
define $C_{i,j} \triangleq \{(1,2),(2,1),(2,3),(3,2),(3,1),(1,3)\}$
for all edge $(i,j)$ of $G$.
Then, any assignment $\psi \colon V \to [3]$ is a \emph{$3$-coloring} of $G$
if and only if $\psi$ satisfies all constraints simultaneously.

\begin{problem}
\label{prb:inapprox:BCSP}
Given a set $V$ of $k$ variables, 
an alphabet set $\Sigma$ of size $n$, and 
a set of constraints $\calC = (C_{i,j})_{i,j \in V}$,
\prb{Binary Constraint Satisfaction Problem} (\prb{BCSP})
asks to find an assignment $\psi \colon V \to \Sigma$ that satisfies
the maximum fraction of constraints.
\end{problem}

It is well known that
\prb{BCSP} parameterized by the number $k$ of variables
is \W{1}-complete from a standard parameterized reduction from \prb{$k$-Clique}.
\citet{lokshtanov2020parameterized} posed a conjecture asserting that 
a constant-factor gap version of \prb{BCSP} is also
\W{1}-hard.

\begin{hypothesis}[Parameterized Inapproximability Hypothesis (PIH) \cite{lokshtanov2020parameterized}]
\label{hyp:inapprox:PIH}
There exists some universal constant $\epsilon \in (0,1)$ such that
it is \W{1}-hard to distinguish between \prb{BCSP} instances
that are promised to either
be satisfiable, or 
have a property that every assignment violates at least
$\epsilon$-fraction of the constraints.
\end{hypothesis}

Here, we prove that 
an \emph{optimization version} of \prb{Grid Tiling} is 
FPT-inapproximable assuming PIH.
Given an instance $\calS = (S_{i,j})_{i,j \in [k]}$ of
\prb{Grid Tiling} and an assignment $\sigma \colon [k]^2 \to [n]^2 $,
$\sigma(i,j)$ and $\sigma(i',j')$ for 
a pair of adjacent cells $(i,j,i',j') \in \calI$
are said to be \emph{consistent}
if they agree on the first coordinate when $i=i'$
or on the second coordinate when $j=j'$, and
\emph{inconsistent} otherwise.
The \emph{consistency} of $\sigma$, denoted $\cons(\sigma)$,
is defined as 
the number of pairs of adjacent cells that are consistent; namely,
\begin{align}
    \cons(\sigma) \triangleq
    \sum_{(i_1,j_1,i_2,j_2) \in \calI}
    \Bigl[\!\!\Bigl[ \sigma(i_1,j_1) \text{ and } \sigma(i_2,j_2) \text{ are consistent} \Bigr]\!\!\Bigr].
\end{align}
The \emph{inconsistency} of $\sigma$ is defined as
the number of inconsistent pairs of adjacent cells.
The optimization version of \prb{Grid Tiling} asks
to find an assignment $\sigma$ such that
$\cons(\sigma)$ is maximized.\footnote{
Our definition is different from \citet{marx2007optimality}
in that
the latter seeks a partial assignment such that the number of defined cells is maximized
while the former requires maximizing the number of consistent adjacent pairs.
}
Note that the maximum possible consistency is $\abs{\calI} = 2k^2$.
We will use $\opt(\calS)$ to denote the optimal consistency
among all possible assignments.
We now demonstrate that \prb{Grid Tiling}
is FPT-inapproximable in an \emph{additive sense} under PIH,
whose proof is reminiscent of \cite{marx2007optimality}.

\begin{lemma}
\label{lem:inapprox:MGT-W1}
Under PIH, there exists some universal constant
$\delta \in (0,1)$ such that
it is \W{1}-hard to distinguish \prb{Grid Tiling} instances between the following cases:
\begin{itemize}
\item Completeness: the optimal consistency is $2k^2$.
\item Soundness: the optimal consistency is at most $2k^2 - \delta k$.
\end{itemize}
\end{lemma}
\begin{proof}
We show a gap-preserving parameterized reduction from
\prb{BCSP} to \prb{Grid Tiling}.
Given an instance of \prb{BCSP} $(V,\Sigma, \calC=(C_{i,j})_{i,j\in V})$,
where $V = [k]$ and $\Sigma = [n]$,
we define $\calS$ to be a collection of
$k^2 $ nonempty subsets $S_{i,j} \subseteq [n]^2$ such that
$S_{i,i} \triangleq [n]^2$ for all $i \in [k]$ and
$S_{i,j} \triangleq C_{i,j} \subseteq [n]^2$ for all $i \neq j \in [k]$.
Suppose first there exists a satisfying assignment $\psi \colon V \to \Sigma$ for the \prb{BCSP} instance; i.e.,
$(\psi(i), \psi(j)) \in C_{i,j}$ for all $i \neq j \in V$.
Constructing another assignment $\sigma \colon [k]^2 \to [n]^2$
for \prb{Grid Tiling} defined by $\calS$
such that $\sigma(i,j) \triangleq (\psi(i), \psi(j)) \in S_{i,j}$
for each $i,j \in [k]$,
we have $\cons(\sigma) = 2k^2$,
proving the completeness part.

Suppose then we are given
an assignment $\sigma \colon [k]^2 \to [n]^2$ for \prb{Grid Tiling} whose inconsistency is
at most $\delta k$ for some $\delta \in (0,1)$.
Define a subset $A \subseteq [n]$ as follows:
\begin{align}
\begin{aligned}
    A \triangleq [k] \setminus \Bigl(
        & \Bigl\{ i\in[k] : \exists j \in [k] \text{ s.t. } \sigma(i,j) \text{ and } \sigma(i,j+1) \text{ are inconsistent} \Bigr\} \cup \\
        & \Bigl\{ j\in[k] : \exists i \in [k] \text{ s.t. } \sigma(i,j) \text{ and } \sigma(i+1,j) \text{ are inconsistent} \Bigr\}
    \Bigr).
\end{aligned}
\end{align}
It follows from the definition that
$\abs{A} \geq (1-\delta) k$ and
the restriction of $\sigma$ on $A^2$ is of zero inconsistency.
Thus, if we define an assignment $\psi: V \to \Sigma$
for \prb{BCSP} as
\begin{align}
    \psi(i) \triangleq z \text{ such that }
    \sigma(i,i) = (z,z) \text{ for each } i \in V,
\end{align}
then it holds that $(\psi(i), \psi(j)) \in S_{i,j} = C_{i,j}$
for all $i \neq j \in A$.
The fraction of constraints in the \prb{BCSP} instance
satisfied by $\psi$ is then at least
\begin{align}
\begin{aligned}
    \frac{{\abs{A} \choose 2}}{{k \choose 2}}
    & = \frac{(k-\delta k)(k-\delta k-1)}{2}\frac{2}{k(k-1)} \\
    & = (1-\delta)\left(1- \frac{\delta k}{k-1}\right) \\
    & \geq (1-\delta)\left(1- \frac{3}{2}\delta\right) & \text{ (for } k \geq 3 \text{)} \\
    & \geq \left(1-\frac{3}{2}\delta\right)^2 \\
    & \geq 1-3\delta.
\end{aligned}
\end{align}
Consequently,
if the optimal consistency $\opt(\calS)$ is
at least $2k^2 - \delta k$ for some $\delta \in (0,1)$, then
the maximum fraction of satisfiable constraints of \prb{BCSP} instance must be at least $1-3\delta$,
which completes (the contraposition of) the soundness part.
Under PIH, it is \W{1}-hard to decide if
the optimal consistency of a \prb{Grid Tiling} instance
is equal to $2k^2$ or less than $2k^2 - \delta k$,
where $\delta \triangleq \frac{\epsilon}{3} $ and
$\epsilon \in (0,1) $ is a constant appearing in \cref{hyp:inapprox:PIH}.
\end{proof}

It should be noted that
we may not be able to 
significantly improve the additive term $\bigO(k)$
owing to a polynomial-time $\epsilon k^2$-additive approximation algorithm for
any constant $\epsilon > 0$:
\begin{observation}
\label{obs:inapprox:MGT-approx}
Given an instance of \prb{Grid Tiling} and 
an error tolerance parameter $\epsilon > 0$,
we can find an assignment whose consistency
is at least $\opt(\calS) - \epsilon k^2$ in
$\epsilon^2 k^2 n^{\bigO(1/\epsilon^2)}$ time.
\end{observation}
\begin{proof}
Given an instance $\calS = (S_{i,j})_{i,j \in [k]}$ of \prb{Grid Tiling}  and $\epsilon > 0$,
if $\epsilon k < 4$, we can use
a brute-force search algorithm to find an optimal assignment
in time $n^{\bigO(k^2)} = n^{\bigO(1/\epsilon^2)}$.
Hereafter, we safely assume that $ \epsilon k \geq 4 $.
Defining
$\displaystyle \ell \triangleq \left\lfloor \frac{\epsilon k}{2}  - 1\right\rfloor$ and
$\displaystyle B \triangleq \left\lceil \frac{k}{\ell} \right\rceil$,
we observe that
\begin{align}
1 \leq \ell \leq \frac{\epsilon k}{2} \text{ and }
B \leq \left\lceil \frac{k}{\frac{\epsilon k}{2} - 1} \right\rceil
\leq \left\lceil \frac{k}{\frac{\epsilon k}{4}} \right\rceil
= \left\lceil \frac{4}{\epsilon} \right\rceil.
\end{align}
We partition $k^2$ cells of $\calS$ into $\ell^2$ blocks,
denoted
$\{P_{\hati, \hatj}\}_{\hati, \hatj \in [\ell]}$,
each of which is of size at most $B^2$ and defined as follows:
\begin{align}
    P_{\hati,\hatj} \triangleq
    \Bigl[B(\hati-1)+1 \isep \min\{B\hati, k\}\Bigr] \times
    \Bigl[B(\hatj-1)+1 \isep \min\{B\hatj, k\}\Bigr]
    \text{ for all } \hati, \hatj \in [\ell].
\end{align}
Consider for each $\hati,\hatj \in [\ell]$,
a \emph{variant} of \prb{Grid Tiling} denoted by
$\calS_{\hati,\hatj} \triangleq \{ S_{i,j} : (i,j) \in P_{\hati,\hatj} \}$, which requires maximizing the number of consistent pairs of
adjacent cells of $P_{\hati,\hatj}$ in $\calI$, 
where $\calI$ is defined in \cref{eq:rank:calI}.
Because each instance $\calS_{\hati,\hatj}$ contains at most $B^2$ cells,
we can solve this variant exactly
by exhaustive search in $n^{B^2 + \bigO(1)}$ time.
Denote by $\sigma_{\hati,\hatj} \colon P_{\hati,\hatj} \to [n]^2$ the obtained partial assignment on $P_{\hati, \hatj}$.
Concatenating all $\sigma_{\hati,\hatj}$'s,
we can construct an assignment $\sigma \colon [k]^2 \to [n]^2$ to
the original \prb{Grid Tiling} instance $\calS$.
Because each partial assignment $\sigma_{\hati,\hatj}$
is optimal on $\calS_{\hati,\hatj}$,
the number of consistent pairs of
adjacent cells within the \emph{same} block
is at least $\opt(\calS)$.
By contrast,
the number of (possibly inconsistent) pairs of
adjacent cells \emph{across} different blocks is $2 k \ell$.
Accordingly, the consistency of $\sigma$ is
$\cons(\sigma) \geq \opt(\calS) - 2 k \ell \geq \opt(\calS) - \epsilon k^2$.
Note that the entire time complexity is bounded by
$\ell^2 \cdot n^{B^2 + \bigO(1)} = \epsilon^2 k^2 n^{\bigO(1/\epsilon^2)}$
completing the proof.
\end{proof}

Our technical result is a gap-preserving parameterized reduction from \prb{Grid Tiling} to \prb{Determinant Maximization},
whose proof is presented in the subsequent subsection.

\begin{lemma}
\label{lem:inapprox:reduction}
There is a polynomial-time, parameterized reduction from
an instance $\calS = (S_{i,j})_{i,j \in [k]}$ of \prb{Grid Tiling} to
an instance $(\mat{A},k^2)$ of \prb{Determinant Maximization} such that 
all diagonal entries of $\mat{A}$ are $1$ and
the following conditions are satisfied:
\begin{itemize}
\item Completeness: If $\opt(\calS) = 2k^2$, then $\maxdet(\mat{A},k^2) = 1$.
\item Soundness: If $\opt(\calS) \leq 2k^2 - \delta k$ for some $\delta > 0$,
then $\maxdet(\mat{A},k^2) \leq 0.999^{\delta k}$.
\end{itemize}
\end{lemma}
Using \cref{lem:inapprox:reduction}, we can prove \cref{thm:inapprox}.
\begin{proof}[Proof of \cref{thm:inapprox}]
Our gap-preserving parameterized reduction is as follows.
Given an instance $\calS = (S_{i,j})_{i,j \in [k]}$ of \prb{Grid Tiling},
we construct an instance $(\mat{A}\in \bbQ^{N \times N}, K \triangleq k^2)$ of \prb{Determinant Maximization} in polynomial time according to \cref{lem:inapprox:reduction},
where $N \triangleq \sum_{i,j \in [k]}\abs{S_{i,j}}$.
The diagonal entries of $\mat{A}$ are $1$ by definition.
Since $K$ is a function only in $k$, this is a parameterized reduction.
According to \cref{lem:inapprox:MGT-W1,lem:inapprox:reduction},
it is \W{1}-hard to determine whether
$\maxdet(\mat{V},K) = 1$ or
$\maxdet(\mat{V},K) \leq 0.999^{\delta k}$ under PIH,
where $\delta \in (0,1)$ is a constant appearing \cref{lem:inapprox:MGT-W1}.
In particular,
\prb{Determinant Maximization} is \W{1}-hard to approximate within a factor better than $0.999^{\delta k} = 2^{-c\sqrt{K}}$ when parameterized by $K$,
where $c \in (0,1)$ is some universal constant.
This completes the proof.
\end{proof}

\subsection{Gap-preserving Reduction from {\normalfont \prb{Grid Tiling}} and Proof of \cref{lem:inapprox:reduction}}
\label{subsec:inapprox:reduction}
To prove \cref{lem:inapprox:reduction},
we describe a gap-preserving parameterized reduction from \prb{Grid Tiling} to \prb{Determinant Maximization}.
Before going into its details,
we introduce a convenient gadget due to \citet{civril2013exponential}.

\begin{lemma}[\protect{\citet[Lemma~13]{civril2013exponential}}]
\label{lem:inapprox:civril-vector}
For any positive even integer $\ell$,
we can construct a set of $2^{\ell}$ rational vectors
$\mat{B}^{(\ell)} = \{\vec{b}_1, \ldots, \vec{b}_{2^{\ell}}\}$
of dimension $2^{\ell+1}$ in $\bigO(4^{\ell})$ time
such that the following conditions are satisfied:
\begin{itemize}
    \item Each entry of vectors is either $0$ or $2^{-\frac{\ell}{2}}$; 
    $\|\vec{b}_i\| = 1 $ for all $i \in [2^{\ell}]$.
    \item $\langle \vec{b}_i, \vec{b}_j \rangle = \frac{1}{2}$ for all $i, j\in [2^\ell]$ with $i \neq j$.
    \item $ \langle \vec{b}_i, \overline{\vec{b}_j} \rangle = \frac{1}{2} $ for all $i,j \in [2^\ell]$ with $i \neq j$, where
    $ \overline{\vec{b}_j} \triangleq 2^{-\frac{\ell}{2}} \cdot \vec{1} - \vec{b}_j $.
\end{itemize}
\end{lemma}
By definition of $\mat{B}^{(\ell)}$, we further have the following:
\begin{align}
    \langle \vec{b}_i, \overline{\vec{b}_i} \rangle
    & = 2^{-\frac{\ell}{2}} \langle \vec{1}, \vec{b}_i \rangle
    - \langle \vec{b}_i, \vec{b}_i \rangle = 0, \\
    \begin{split}
    \langle \overline{\vec{b}_i}, \overline{\vec{b}_j} \rangle
    & = \langle 2^{-\frac{\ell}{2}} \vec{1} - \vec{b}_i, 2^{-\frac{\ell}{2}} \vec{1} - \vec{b}_j \rangle \\
    & = 2^{-\ell} \langle \vec{1}, \vec{1} \rangle
    - 2^{-\frac{\ell}{2}} \langle \vec{1}, \vec{b}_i + \vec{b}_j \rangle
    + \langle \vec{b}_i, \vec{b}_j \rangle \\
    & = \langle \vec{b}_i, \vec{b}_j \rangle.
    \end{split}
\end{align}
Our reduction strategy is very similar to that of \cref{thm:rank}.
Given an instance $\calS = (S_{i,j})_{i,j \in [k]}$ of
\prb{Grid Tiling},
we construct a rational vector $\vec{v}_{x,y}^{(i,j)}$ for
each element $(x,y) \in S_{i,j}$ of cell $(i,j) \in [k]^2$.
Each vector consists of $\abs{\calI} = 2k^2$ blocks indexed by $\calI$,
each of which is either a vector in the set $\mat{B}^{(2\lceil \log n \rceil)}$ or the zero vector $\vec{0}$.
Hence, the dimension of the vectors is
$ 2k^2 \cdot 2^{2\lceil \log n \rceil +1} = \bigO(k^2 n^2) $.
Let $\vec{v}_{x,y}^{(i,j)}(i_1,j_1,i_2,j_2)$ denote the block of $\vec{v}_{x,y}^{(i,j)}$
corresponding to a pair of adjacent cells $(i_1,j_1,i_2,j_2) \in \calI$.
Each block is subsequently defined as follows:
\begin{align}
    \vec{v}^{(i,j)}_{x,y}(e) \triangleq
    \begin{dcases}
        \overline{\vec{b}_x} & \text{if } e = (i,j-1,i,j), \\
        \vec{b}_x & \text{if } e = (i,j,i,j+1), \\
        \overline{\vec{b}_y} & \text{if } e = (i-1,j,i,j), \\
        \vec{b}_y & \text{if } e = (i,j,i+1,j), \\
        \vec{0} & \text{otherwise}.
    \end{dcases}
\end{align}
Hereafter,
two vectors $\vec{v}^{(i,j)}_{x,y}$ and $\vec{v}^{(i',j')}_{x',y'}$
are said to be \emph{adjacent} if $(i,j)$ and $(i',j')$ are adjacent, and
two adjacent vectors are said to be \emph{consistent} if
$(x,y)$ and $(x',y')$ are consistent
(i.e., $x=x'$ whenever $i=i'$ and $y=y'$ whenever $j=j'$) and
\emph{inconsistent} otherwise.
Since each vector contains exactly four vectors chosen from $\mat{B}^{(2 \lceil \log n \rceil)}$,
its squared norm is equal to $4$.
In addition,
$\vec{v}^{(i,j)}_{x,y}$ and $\vec{v}^{(i',j')}_{x',y'} $ are orthogonal
whenever $(i,j) $ and $(i',j')$ are not identical or adjacent.
Observe further that
if two cells are adjacent,
the inner product of two vectors in $\mat{V}$ is calculated as follows:
\begin{align}
    \left\langle \vec{v}^{(i,j)}_{x,y}, \vec{v}^{(i,j+1)}_{x',y'} \right\rangle & =
    \langle \vec{b}_x, \overline{\vec{b}_{x'}} \rangle = 
    \begin{cases}
        0 & \text{if they are consistent (i.e., } x = x'\text{)}, \\
        \frac{1}{2} & \text{otherwise (i.e., } x \neq x'\text{)},
    \end{cases}
    \label{eq:inapprox:dot-adjacent-i}
    \\
    \left\langle \vec{v}^{(i,j)}_{x,y}, \vec{v}^{(i+1,j)}_{x',y'} \right\rangle & =
    \langle \vec{b}_y, \overline{\vec{b}_{y'}} \rangle = 
    \begin{cases}
        0 & \text{if they are consistent (i.e., } y = y'\text{)}, \\
        \frac{1}{2} & \text{otherwise (i.e., } y \neq y'\text{)}.
    \end{cases}
    \label{eq:inapprox:dot-adjacent-j}
\end{align}
On the other hand,
the inner product of two vectors in the same cell is as follows:
\begin{align}
\label{eq:inapprox:dot-identical}
\begin{aligned}
    \left\langle \vec{v}^{(i,j)}_{x,y}, \vec{v}^{(i,j)}_{x',y'} \right\rangle
    & =
    2 \cdot \langle \vec{b}_x, \vec{b}_{x'} \rangle
    + 2 \cdot \langle \vec{b}_y, \vec{b}_{y'} \rangle \\
    & =
    \begin{cases}
        4 & \text{if } x=x' \text{ and } y=y', \\
        3 &
        \text{if } x=x' \text{ xor } y=y',
        \\
        2 & \text{if } x \neq x' \text{ and } y \neq y'.
    \end{cases}
\end{aligned}
\end{align}
We denote
by $\mat{V}^{(i,j)}$ the set of vectors corresponding to the elements of $S_{i,j}$; i.e.,
$
    \mat{V}^{(i,j)} \triangleq \{ \vec{v}^{(i,j)}_{x,y} : (x,y) \in S_{i,j} \}
$
for each $i,j \in [k]$.
We now define an instance $(\mat{V}, K)$ of \prb{Determinant Maximization} as
$
    \mat{V} \triangleq \bigcup_{i,j \in [k]} \mat{V}^{(i,j)}
    \text{ and }
    K \triangleq k^2.
$
Note that $\mat{V}$ contains $N \triangleq \sum_{i,j \in [k] }\abs{S_{i,j}}$ vectors.

We now proceed to the proof of (the soundness argument of) \cref{lem:inapprox:reduction}.
Let $\mat{S}$ be a set of $k^2$ vectors from $\mat{V}$.
Define
$\mat{S}^{(i,j)} \triangleq \mat{V}^{(i,j)} \cap \mat{S} = \{ \vec{v}^{(i,j)}_{x,y} \in \mat{S} : (x,y) \in S_{i,j} \}$
for each $i,j\in [k]^2$.
Denote by $\cov(\mat{S})$
the number of cells $(i,j) \in [k]^2$ such that
$\mat{S}$ includes $\vec{v}^{(i,j)}_{x,y}$ for some $(x,y)$; i.e.,
$
    \cov(\mat{S})
    \triangleq \{ (i,j) \in [k]^2 : \mat{S}^{(i,j)} \neq \emptyset \},
$
and we also define 
$
    \dup(\mat{S})
    \triangleq \{ (i,j) \in [k]^2 : \mat{S}^{(i,j)} = \emptyset \}.
$
It follows from the definition that 
$\cov(\mat{S}) + \dup(\mat{S}) = k^2$ and
$\dup(\mat{S})$
counts the total number of ``duplicate'' vectors in the same cell.
We first present an upper bound on
the volume of $\mat{S}$ in terms of $\dup(\mat{S})$,
implying that we cannot select many duplicate vectors from the same cell.
\begin{lemma}
\label{lem:inapprox:volume-rep}
If $\dup(\mat{S}) \leq \frac{k^2}{2}$,
then it holds that 
\begin{align}
    \vol^2(\mat{S}) \leq 4^{k^2} \cdot \left(\frac{3}{4}\right)^{\dup(\mat{S})}.
\end{align}
\end{lemma}
\begin{proof}
We first introduce \emph{Fischer's inequality}.
Suppose $\mat{A}$ and $\mat{B}$ are respectively $m \times m$ and $n \times n$
positive semi-definite matrices, and $\mat{C}$ is an $m \times n$ matrix.
Then, it holds that
$
    \det\left(\begin{bsmallmatrix}
    \mat{A} & \mat{C} \\
    \mat{C}^\top & \mat{B}
    \end{bsmallmatrix}\right)
    \leq \det(\mat{A}) \cdot \det(\mat{B}).
$
As a corollary, we have
the volume version of Fischer's inequality stating that
for any two sets of vectors $\mat{P}, \mat{Q}$,
\begin{align}
\label{eq:inapprox:fischer-volume}
    \vol(\mat{P} \uplus \mat{Q}) \leq \vol(\mat{P}) \cdot \vol(\mat{Q}).
\end{align}
Because we have
\begin{align}
    \vol(\mat{S}) \leq \prod_{i,j \in [k]} \vol(\mat{S}^{(i,j)})
\end{align}
by Fischer's inequality in \cref{eq:inapprox:fischer-volume},
we consider bounding $\vol(\mat{S}^{(i,j)})$ from above for each $i,j \in [k]$.
We will show the following:
\begin{align}
\label{eq:inapprox:volume-rep:claim}
    \vol^2(\mat{S}^{(i,j)}) = 
    \begin{cases}
        1 & \text{if } \mat{S}^{(i,j)} = \emptyset, \\
        4 \cdot 3^{\abs{\mat{S}^{(i,j)}}-1} & \text{otherwise}.
    \end{cases}
\end{align}
Suppose $\mat{S}^{(i,j)} = \{\vec{v}_1, \ldots, \vec{v}_m\}$
for $m \triangleq \abs{\mat{S}^{(i,j)}} $.
By the definition of volume in \cref{eq:def-vol}, we have
\begin{align}
\begin{aligned}
    \vol^2(\mat{S}^{(i,j)}) & = 
    \|\vec{v}_1\|^2 \cdot \prod_{2 \leq i \leq m}
    \Bigl\| \vec{v}_i - \proj_{\{\vec{v}_1, \ldots, \vec{v}_{i-1}\}}(\vec{v}_i) \Bigr\|^2 \\
    & \leq 4 \cdot \prod_{2 \leq i \leq m}
    \Bigl\| \vec{v}_i - \proj_{\{\vec{v}_1\}}(\vec{v}_i) \Bigr\|^2.
\end{aligned}
\end{align}
Because the projection of each $\vec{v}_i$ with $i \neq 1$ onto $\vec{v}_1$ is calculated as
\begin{align}
    \proj_{\{\vec{v}_1\}}(\vec{v}_i) = 
    \frac{\langle \vec{v}_1, \vec{v}_i \rangle}{\|\vec{v}_1\|^2} \vec{v}_1
    = \frac{\langle \vec{v}_1, \vec{v}_i \rangle}{4} \vec{v}_1,
\end{align}
we obtain
\begin{align}
\begin{aligned}
    \Bigl\|\vec{v}_i - \proj_{\{\vec{v}_1\}}(\vec{v}_i)\Bigr\|^2
    & = \|\vec{v}_i\|^2 + \left\| \frac{\langle \vec{v}_1, \vec{v}_i \rangle}{4} \vec{v}_1 \right\|^2
    - 2 \cdot \left\langle \vec{v}_i, \frac{\langle \vec{v}_1, \vec{v}_i \rangle}{4} \vec{v}_1 \right\rangle
    \\
    & = 4 + \frac{\langle \vec{v}_1, \vec{v}_i \rangle^2}{16} \cdot 4
    - \frac{\langle \vec{v}_1, \vec{v}_i \rangle^2}{2} \\
    & = 4 - \frac{1}{4} \cdot \langle \vec{v}_1, \vec{v}_i \rangle^2.
\end{aligned}
\end{align}
By \cref{eq:inapprox:dot-identical},
$\langle \vec{v}_1, \vec{v}_i \rangle$ with $i \neq 1$ is either $2$ or $3$;
hence, it holds that $\Bigl\|\vec{v}_i - \proj_{\{\vec{v}_1\}}(\vec{v}_i)\Bigr\|^2 \leq 3$, ensuring \cref{eq:inapprox:volume-rep:claim} as desired.
Consequently, we get 
\begin{align}
\begin{aligned}
    \vol^2(\mat{S})
    \leq \prod_{i,j \in [k]} \vol^2(\mat{S}^{(i,j)})
    \leq \prod_{i,j \in [k]: \mat{S}^{(i,j)} \neq \emptyset}
    4 \cdot 3^{\abs{\mat{S}^{(i,j)}} - 1} \\
    = \left(\frac{4}{3}\right)^{\cov(\mat{S})} \cdot \prod_{i,j \in [k]} 3^{\abs{\mat{S}^{(i,j)}}} 
    = 4^{k^2} \cdot \left(\frac{3}{4}\right)^{\dup(\mat{S})},
\end{aligned}
\end{align}
completing the proof.
\end{proof}

We then present another upper bound on the volume of $\mat{S}$
in terms of the \emph{inconsistency} of a partial solution of \prb{Grid Tiling} constructed from the selected vectors.
For a set $\mat{S}$ of $k^2$ vectors from $\mat{V}$,
a \emph{partial} assignment
$\sigma_{\mat{S}} \colon [k]^2 \to [n]^2 \cup \{\bigstar\}$ for \prb{Grid Tiling} is defined as
\begin{align}
    \sigma_{\mat{S}}(i,j) \triangleq
    \begin{cases}
        \text{any } (x,y) \text{ s.t.~} \vec{v}^{(i,j)}_{x,y} \in \mat{S}^{(i,j)} & \text{if such } (x,y) \text{ exists}, \\
        \bigstar  & \text{otherwise (i.e., } \mat{S}^{(i,j)} = \emptyset \text{)},
    \end{cases}
\end{align}
where the symbol ``$\bigstar$'' means \emph{undefined} and
the choice of $(x,y)$ is arbitrary.
The inconsistency of a partial assignment $\sigma_{\mat{S}}$ is defined as 
\begin{align}
    \sum_{(i_1,j_1,i_2,j_2) \in \calI}
    \left[\!\!\left[
    \begin{gathered}
    \sigma(i_1,j_1) \neq \bigstar \text{; } \sigma(i_2,j_2) \neq \bigstar \text{;} \\
    \sigma(i_1,j_1)\text{ and }\sigma(i_2,j_2)\text{ are inconsistent}
    \end{gathered}
    \right]\!\!\right].
\end{align}
Note that the sum of the consistency and inconsistency of $\sigma_{\mat{S}}$
is no longer necessarily $2k^2$.
Using $\sigma_{\mat{S}}$, we define a partition $(\mat{P}, \mat{Q})$ of $\mat{S}$ as
$\mat{P} \triangleq \{ \vec{v}^{(i,j)}_{x,y} \in \mat{S} : i,j \in [k], \sigma_{\mat{S}}(i,j) = (x,y) \}$ and
$\mat{Q} \triangleq \mat{S} \setminus \mat{P}$.
We further prepare an arbitrary \emph{ordering} $\prec$
over $[k]^2$;
e.g., $(i,j) \prec (i',j')$ if $i<i'$, or $i=i'$ and $j<j'$.
We abuse the notation by writing
$\vec{v}^{(i,j)}_{x,y} \prec \vec{v}^{(i',j')}_{x',y'}$ for any two vectors of $\mat{V}$ whenever $(i,j) \prec (i',j')$.
Define now
$ 
    \mat{P}_{\prec \vec{v}} \triangleq \{ \vec{u} \in \mat{P} : \vec{u} \prec \vec{v} \}.
$
The following lemma states that 
the squared volume of $k^2$ vectors
exponentially decays in the minimum possible inconsistency among all assignments of $\calS$.

\begin{lemma}
\label{lem:inapprox:volume-incons}
Suppose $\opt(\calS) \leq 2k^2 - \delta k$ for some $\delta > 0$ and
$\cov(\mat{S}) \geq k^2 - \gamma k$ for some $\gamma > 0$.
If $\delta k - 4 \gamma k$ is positive,
then it holds that
\begin{align}
    \vol^2(\mat{S}) \leq 4^{k^2} \cdot \left(\frac{63}{64}\right)^{\frac{\delta - 4 \gamma}{4} k}.
\end{align}
\end{lemma}
The proof of \cref{lem:inapprox:volume-incons} involves the following claim.
\begin{claim}
\label{lem:inapprox:incons}
Suppose the same conditions as in \cref{lem:inapprox:volume-incons} are satisfied.
Then,
the inconsistency of $\sigma_{\mat{S}}$ is at least $\delta k - 4\gamma k$.
Moreover, the number of vectors $\vec{v}$ in $\mat{P}$ such that
$\vec{v}$ is inconsistent with some adjacent vector of $\mat{P}_{\prec \vec{v}}$
is at least
$ \frac{\delta k - 4\gamma k}{4}$.
\end{claim}
\begin{claimproof}
Assuming that
the inconsistency of $\sigma_{\mat{S}}$ is less than $\delta k - 4 \gamma k$,
we can construct a (complete) assignment
$\sigma \colon [k]^2 \to [n]^2$ from $\sigma_{\mat{S}}$
by filling in undefined values (i.e., ``$\bigstar$'')
with an arbitrary integer pair in the corresponding cell.
The inconsistency of $\sigma$ is clearly less than $\delta k$, which is a contradiction.

On the other hand,
assume that there are only less than $\frac{\delta k - 4\gamma k}{4}$ vectors $\vec{v}$ in $\mat{P}$ that are inconsistent with some adjacent vector of $\mat{P}_{\prec \vec{v}}$.
Then, the inconsistency of $\sigma_{\mat{S}}$ must be less than $\frac{\delta k - 4\gamma k}{4} \cdot 4$,\footnote{
Note that double counting of inconsistent pairs of adjacent cells does not occur.
}
which is a contradiction.
\end{claimproof}

\begin{proof}[Proof of \cref{lem:inapprox:volume-incons}]
By applying \cref{eq:inapprox:fischer-volume}, we have
$\vol^2(\mat{S}) \leq \vol^2(\mat{P}) \cdot \vol^2(\mat{Q})$.
Observe easily that
\begin{align}
    \vol^2(\mat{Q}) \leq 4^{\abs{\mat{Q}}} \text{ and }
    \vol^2(\mat{P}) = \prod_{\vec{v} \in \mat{P}}
    \dis^2(\vec{v}, \mat{P}_{\prec \vec{v}}).
\end{align}
Thereafter, we bound $\dis^2(\vec{v}, \mat{P}_{\prec \vec{v}})$
for each $\vec{v} \in \mat{P}$ using the following case analysis:
\begin{description}
\item[\textbf{Case 1.}] $\vec{v}$ is consistent with all adjacent vectors in $\mat{P}_{\prec \vec{v}}$:
because $\langle \vec{u}, \vec{v} \rangle = 0$
for all $\vec{u} \in \mat{P}_{\prec \vec{v}}$,
we have $\proj_{\mat{P}_{\prec \vec{v}}}(\vec{v}) = \vec{0}$,
implying that
\begin{align}
    \dis^2(\vec{v}, \mat{P}_{\prec \vec{v}}) =
    \Bigl\|\vec{v} - \proj_{\mat{P}_{\prec \vec{v}}}(\vec{v})\Bigr\|^2 = \|\vec{v}\|^2 = 4.
\end{align}
\item[\textbf{Case 2.}] There exists a vector $\vec{u} $ in $\mat{P}_{\prec \vec{v}}$
that is inconsistent with $\vec{v}$:
because $\langle \vec{u}, \vec{v} \rangle = \frac{1}{2}$ by \cref{eq:inapprox:dot-adjacent-i,eq:inapprox:dot-adjacent-j},
the projection of $\vec{v}$ on $\vec{u}$ is equal to
\begin{align}
    \proj_{\{\vec{u}\}}(\vec{v}) = \frac{\langle \vec{u}, \vec{v} \rangle}{\|\vec{u}\|^2} \cdot \vec{u} = \frac{1}{8} \cdot \vec{u}.
\end{align}
Therefore, we have
\begin{align}
\begin{aligned}
    \dis^2(\vec{v}, \mat{P}_{\prec \vec{v}})
    & \leq \dis^2(\vec{v}, \{\vec{u}\})
    = \Bigl\| \vec{v} - \proj_{\{\vec{u}\}}(\vec{v}) \Bigr\|^2 \\
    & = \|\vec{v}\|^2 + \left\|\frac{1}{8} \cdot \vec{u}\right\|^2 - 2 \cdot \frac{1}{8} \langle \vec{u}, \vec{v} \rangle
    \\
    & = 4 + \frac{4}{64} - 2 \cdot \frac{1}{16}
    = \frac{63}{16}.
\end{aligned}
\end{align}
\end{description}
As illustrated in \cref{lem:inapprox:incons},
at least $\frac{\delta- 4 \gamma}{4}k$ vectors of $\mat{P}$ fall into the latter case; it thus turns out that
\begin{align}
    \vol^2(\mat{P})
    = \prod_{\vec{v} \in \mat{P}} \dis^2(\vec{v}, \mat{P}_{\prec \vec{v}})
    \leq 4^{\abs{\mat{P}} - \frac{\delta - 4 \gamma}{4}k} \cdot \left(\frac{63}{16}\right)^{\frac{\delta - 4 \gamma}{4}k}
    = 4^{\abs{\mat{P}}} \cdot \left(\frac{63}{64}\right)^{\frac{\delta - 4 \gamma}{4}k}.
\end{align}
Consequently, the squared volume of $\mat{S}$ can be bounded as
\begin{align}
    \vol^2(\mat{S}) \leq 4^{\abs{\mat{P}} + \abs{\mat{Q}}} \cdot \left(\frac{63}{64}\right)^{\frac{\delta - 4 \gamma}{4}k}
    = 4^{k^2} \cdot \left(\frac{63}{64}\right)^{\frac{\delta - 4 \gamma}{4}k},
\end{align}
which completes the proof.
\end{proof}

Using \cref{lem:inapprox:volume-rep,lem:inapprox:volume-incons},
we can easily conclude the proof of \cref{lem:inapprox:reduction} as follows.
\begin{proof}[Proof of \cref{lem:inapprox:reduction}]
Observe that the reduction described in \cref{subsec:inapprox:reduction}
is a parameterized reduction as
it requires polynomial time and
an instance $\calS = (S_{i,j})_{i,j \in [k]}$ of \prb{Grid Tiling} is transformed into
an instance $(\mat{V}, k^2)$ of \prb{Determinant Maximization}.
In addition, the construction of $\mat{B}^{(2\lceil \log n \rceil)}$
completes in time $\bigO(4^{2 \lceil \log n \rceil}) = \bigO(n^4)$ by \cref{lem:inapprox:civril-vector}.

We now prove the correctness of the reduction.
Let us begin with the completeness argument.
Suppose $\opt(\calS) = 2k^2$; i.e.,
there is an assignment $\sigma$ of consistency $2k^2$.
Then, $k^2$ vectors in the set
$\mat{S} \triangleq \{ \vec{v}^{(i,j)}_{\sigma(i,j)} : i,j \in [k] \}$
are orthogonal to each other,
implying that $\vol^2(\mat{S}) = 4^{k^2}$.
On the other hand, because every vector of $\mat{V}$ is of
squared norm $4$,
the maximum possible squared volume among $k^2$ vectors in $\mat{V}$ is
$4^{k^2}$; namely, $\maxdet(\mat{V},k^2) = 4^{k^2}$.

We then prove the soundness argument.
Suppose $\opt(\calS) \leq 2k^2 - \delta k$ for some constant $\delta > 0$.
Then, for any set $\calS$ of $k^2$ vectors from $\mat{V}$ such that
$
    \dup(\mat{S}) > \frac{\log 0.999^{-1}}{\log (\frac{3}{4})^{-1}} \cdot \delta k,
$
we have that by \cref{lem:inapprox:volume-rep},
$
    \vol^2(\mat{S}) < 4^{k^2} \cdot 0.999^{\delta k}.
$
It is thus sufficient to consider the case that
\begin{align}
    \dup(\mat{S}) \leq \frac{\log 0.999^{-1}}{\log (\frac{3}{4})^{-1}} \cdot \delta k \approx 0.0035 \cdot \delta.
\end{align}
In particular,
it suffices to assume that $\dup(\mat{S}) \leq \gamma k$ for some $\gamma \in (0, \frac{\delta}{4})$.
Simple calculation using \cref{lem:inapprox:volume-rep,lem:inapprox:volume-incons} derives that
\begin{align}
\begin{aligned}
    \vol^2(\mat{S})
    & \leq \min\left\{
    4^{k^2} \cdot \left(\frac{3}{4}\right)^{\gamma k},
    4^{k^2} \cdot \left(\frac{63}{64}\right)^{\frac{\delta - 4\gamma }{4}k} \right
    \} \\
    & \leq 4^{k^2} \cdot \min\left\{ \left(\frac{63}{64}\right)^{\gamma k}, \left(\frac{63}{64}\right)^{\frac{\delta - 4\gamma}{4}k} \right\} \\
    & \leq 4^{k^2} \cdot \left(\frac{63}{64}\right)^{\displaystyle
    \underbrace{\min\left\{\gamma, \frac{\delta - 4\gamma}{4}\right\}}_{\heartsuit}
    \cdot k } \\
    & \leq 4^{k^2} \cdot \left(\frac{63}{64}\right)^{\frac{\delta}{8} k}
    \leq 4^{k^2} \cdot 0.999^{\delta k},
\end{aligned}
\end{align}
where the second-to-last inequality is due to the fact that 
$\heartsuit$ is maximized when $ \gamma = \frac{\delta - 4 \gamma}{4} $;
i.e., $\gamma = \frac{\delta}{8} > 0$.

Because the diagonal entries of the Gram matrix $\mat{A}$ defined by the vectors of $\mat{V}$ are $4$,
we can construct another instance of \prb{Determinant Maximization} as
$(\widetilde{\mat{A}}, k^2)$,
where $\widetilde{\mat{A}} \triangleq \frac{1}{4}\mat{A}$.
Observe finally that the diagonal entries of $\widetilde{\mat{A}}$ are $1$ and
$\det(\widetilde{\mat{A}}_S) = 4^{-\abs{S}} \cdot \det(\mat{A}_S)$ for any $S$,
which completes the proof.
\end{proof}

\subsection{$\epsilon$-Additive FPT-Approximation Parameterized by Rank}
\label{subsec:inapprox:rank}
Here, we develop
an $\epsilon$-additive FPT-approximation algorithm 
parameterized by the \emph{rank} of an input matrix $\mat{A}$,
provided that $\mat{A}$ is the Gram matrix of vectors of \emph{infinity norm at most~$1$}.
Our algorithm complements \cref{lem:inapprox:reduction} in a sense that
we can solve the promise problem in FPT time with respect to $\rank(\mat{A})$.
The proof uses the standard rounding technique.

\begin{observation}
\label{obs:approx}
Let $\vec{v}_1, \ldots, \vec{v}_n$ be $n$ $d$-dimensional vectors in $\bbQ^d$
such that $\|\vec{v}_i\|_{\infty} \leq 1$
for all $i \in [n]$,
$\mat{A}$ the Gram matrix defined by the vectors,
$k \in [d]$ a positive integer, and
$\epsilon > 0$ an error tolerance parameter.
Then, we can compute an approximate solution $S \in {[n] \choose k}$ to
\prb{Determinant Maximization} in
$\epsilon^{-d^2} \cdot d^{\bigO(d^3)} \cdot n^{\bigO(1)}$ time
such that $\det(\mat{A}_S) \geq \maxdet(\mat{A},k) - \epsilon$.
\end{observation}
\begin{proof}
Let $\vec{v}_1, \ldots, \vec{v}_n$ be $n$ $d$-dimensional vectors in $\bbQ^d$
such that $\|\vec{v}_i\|_{\infty} \leq 1$ for all $i \in [n]$, and
$\mat{A}$ be the Gram matrix in $\bbQ^{n \times n}$ defined by them; i.e.,
$A_{i,j} = \langle \vec{v}_i, \vec{v}_j \rangle$ for all $i,j \in [n]$.
We introduce a parameter $\haba > 0$,
which is a reciprocal of some positive integer.
The value of $\haba$ is determined later based on $d$, $k$, and $\epsilon$.
We then define
the set $I_{\haba}$ of rational numbers equally spaced on the interval $[-1,1]$ as follows:
\begin{align}
    I_{\haba} \triangleq
    \Bigl\{-1, -1 + \haba, \ldots, -2\haba, -\haba, 0, \haba, 2\haba, \ldots, 1-\haba, 1\Bigr\}.
\end{align}
Subsequently, we construct $n$ $d$-dimensional vectors $\vec{w}_1, \ldots, \vec{w}_n$ from $\vec{v}_1, \ldots, \vec{v}_n$ as follows:
for each $i \in [n]$ and $e \in [d]$,
$w_i(e)$ is defined to be a number of $I_\haba$ closest to $v_i(e)$.
Observe that $\abs{v_i(e) - w_i(e)} \leq \haba$ and $\|\vec{w}_i\|_\infty \leq 1$.
Let $\mat{B} \in \bbQ^{n \times n}$ be the Gram matrix defined by $\vec{w}_1, \ldots, \vec{w}_n$.
The absolute difference of the determinant between $\mat{A}$ and $\mat{B}$
is bounded as shown below,
whose proof is based on the application of \cref{lem:perturbation}.

\begin{claim*}
For $n$ $d$-dimensional vectors $\vec{v}_1, \ldots, \vec{v}_n$
such that $\|\vec{v}_i\|_{\infty} \leq 1$ for all $i \in [n]$ and
$n$ $d$-dimensional vectors $\vec{w}_1, \ldots, \vec{w}_n$ constructed from $\vec{v}_i$'s and $\haba$ according to the procedure described above,
let 
$\mat{A}$ and $\mat{B}$ be the Gram matrices defined respectively by
$\vec{v}_i$'s and $\vec{w}_i$'s.
Then, for any set $S \in {[n] \choose k}$ for $k \leq d$,
it holds that,
\begin{align}
    \abs{\det(\mat{A}_S) - \det(\mat{B}_S)} \leq 3 \cdot d^{2d+1} \cdot \haba.
\end{align}
Moreover, the number of distinct vectors in the set
$\{ \vec{w}_1, \ldots, \vec{w}_n \}$
is at most
$(\frac{2}{\haba} + 1)^d$.
\end{claim*}
\begin{claimproof}
To apply \cref{lem:perturbation}, we first bound the matrix norm of $\mat{A}$, $\mat{B}$, and $\mat{A}-\mat{B}$.
The max norm of $\mat{A}$ and $\mat{B}$ can be bounded as follows:
\begin{align}
    \|\mat{A}\|_{\max} & = \max_{i,j} \abs{A_{i,j}}
    \leq \max_{i,j} \left\{\sum_{e \in [d]} \abs{v_i(e) \cdot v_j(e)}\right\} \leq d, \\
    \|\mat{B}\|_{\max} & = \max_{i,j} \abs{B_{i,j}}
    \leq \max_{i,j} \left\{\sum_{e \in [d]} \abs{w_i(e) \cdot w_j(e)}\right\} \leq d,
\end{align}
where we used the fact that
$\|\vec{v}_i\|_{\infty} \leq 1$ and $\|\vec{w}_i\|_{\infty} \leq 1$.
For each $i \in [n]$ and $e \in [d]$,
let $w_i(e) = v_i(e) + \haba_{i,e}$ for some $\abs{\haba_{i,e}} \leq \haba$.
We then bound the absolute difference between
$v_i(e) \cdot v_j(e)$ and $w_i(e) \cdot w_j(e)$
for each $i,j \in [n]$ and $e \in [d]$ as
\begin{align}
\begin{aligned}
    & \abs{v_i(e) \cdot v_j(e) - w_i(e) \cdot w_j(e)} \\
    & = \abs{ v_i(e)\cdot v_j(e) - (v_i(e)+\haba_{i,e})\cdot (v_j(e)+\haba_{j,e}) } \\
    & = \abs{\haba_{i,e} \cdot v_j(e) + \haba_{j,e} \cdot v_i(e) + \haba_{i,e} \cdot \haba_{j,e}} \\
    & \leq 2\haba + \haba^2 \leq 3 \haba,
\end{aligned}
\end{align}
where we used the fact that $\haba_{i,e} \leq \haba \leq 1$.
Consequently, the max norm of the difference between $\mat{A}$ and $\mat{B}$ can be bounded as
\begin{align}
\begin{aligned}
\|\mat{A}-\mat{B}\|_{\max}
& = \max_{i,j} \abs{A_{i,j} - B_{i,j}} \\
& = \max_{i,j} \abs{ \langle \vec{v}_i, \vec{v}_j \rangle - \langle \vec{w}_i, \vec{w}_j \rangle } \\
& \leq \sum_{e \in [d]} \abs{v_i(e) \cdot v_j(e) - w_i(e)\cdot w_j(e)} \\
& \leq 3d \haba.
\end{aligned}
\end{align}
Calculation using \cref{lem:perturbation} derives that for any set $S \in {[n] \choose k}$ for $k \leq d$,
\begin{align}
\begin{aligned}
& \abs{\det(\mat{A}_S) - \det(\mat{B}_S)} \\
& \leq k \cdot \max\{ \|\mat{A}_S\|_2, \|\mat{B}_S\|_2 \}^{k-1} \cdot \|\mat{A}_S - \mat{B}_S\|_2 \\
& \leq k \cdot \max\{ k \cdot \|\mat{A}\|_{\max}, k \cdot \|\mat{B}\|_{\max} \}^{k-1} \cdot k \cdot \|\mat{A} - \mat{B}\|_{\max} \\
& \leq k \cdot (kd)^{k-1} \cdot k \cdot 3d \haba \\
& \leq 3 \cdot d^{2d+1} \cdot \haba.
\end{aligned}
\end{align}
Since every vector $\vec{w}_i$ is in $I_{\haba}^d$,
the number of distinct vectors in the set
$\{\vec{w}_1, \ldots, \vec{w}_n\}$
is bounded by $\abs{I_{\haba}^d} = (\frac{2}{\haba} + 1)^d$,
completing the proof.
\end{claimproof}

Our parameterized algorithm works as follows.
Given the Gram matrix $\mat{A} \in \bbQ^{n \times n}$ of
$n$ $d$-dimensional rational vectors $\vec{v}_1, \ldots, \vec{v}_n \in \bbQ^d$ and an error tolerance parameter $\epsilon > 0$,
we set the value of $\haba$ as
\begin{align}
\haba \triangleq \frac{\lfloor \epsilon^{-1} \rfloor^{-1}}{6 \cdot d^{2d+1}} = \epsilon \cdot d^{\bigO(d)}.
\end{align}
Since $1/\haba$ is a positive integer by definition, $I_{\haba}$ is defined;
we construct $n$ $d$-dimensional rational vectors
$\mat{W} \triangleq \{\vec{w}_1, \ldots, \vec{w}_n\} $ in $\bbQ^d$
according to the procedure described above and
the Gram matrix $\mat{B} \in \bbQ^{n \times n}$ defined by $\mat{W}$.
We claim that \prb{Determinant Maximization} on $\mat{B}$
can be solved exactly in FPT time with respect to $\epsilon^{-1}$ and $d$.
Observe that if $\mat{W}$ includes at most $k-1$ distinct vectors,
we can return an arbitrary set $S \in {[n] \choose k}$ whose principal minor is always $\det(\mat{B}_S) = 0$;
therefore, we can safely consider the case $\abs{\mat{W}} \geq k$ only.
According to the claim above, it holds that $\abs{\mat{W}} \leq (\frac{2}{\haba} + 1)^d$.
We can thus enumerate the set of $k$ distinct vectors in $\mat{W}$ in time
\begin{align}
    {(\frac{2}{\haba} + 1)^d \choose k} \leq \left(\frac{2}{\haba} + 1\right)^{dk} \leq \left(\frac{1}{\epsilon}\right)^{d^2} \cdot d^{\bigO(d^3)}.
\end{align}
Calculating the principal minor in $\bigO(k^3)$ time for
each Gram matrix defined by $k$ distinct vectors,
we can find the one,
denoted $S_B^*$, having the maximum principal minor and return $S_B^*$ as a solution.
The overall time complexity is bounded by
$(1/\epsilon)^{d^2} \cdot d^{\bigO(d^3)} \cdot \bigO(k^3) \cdot n^{\bigO(1)} = \epsilon^{-d^2} \cdot d^{\bigO(d^3)} \cdot n^{\bigO(1)}$.
We finally prove an approximation guarantee of $S_B^*$.
Let $S^*$ be the optimal solution for \prb{Determinant Maximization} on $\mat{A}$; i.e., $S^* \triangleq \argmax_{S \in {[n] \choose k}} \det(\mat{A}_S)$.
By the claim above, we have for any set $S \in {[n] \choose k}$,
\begin{align}
\label{eq:approx:detA-detB}
    \abs{\det(\mat{A}_S) - \det(\mat{B}_S)} \leq 3\cdot  d^{2d+1}\cdot \haba 
    \leq 3 \cdot d^{2d+1} \cdot \frac{\epsilon}{6 \cdot d^{2d+1}} = \frac{\epsilon}{2}.
\end{align}
In particular, it holds that
\begin{align}
\begin{aligned}
    \det(\mat{A}_{S_B^*}) & \geq \det(\mat{B}_{S_B^*}) - \frac{\epsilon}{2}
    & \text{(by \cref{eq:approx:detA-detB})} \\
    & \geq \det(\mat{B}_{S^*}) - \frac{\epsilon}{2}
    & \text{(by optimality of } S_B^* \text{)} \\
    & \geq \det(\mat{A}_{S^*}) - \epsilon.
    & \text{(by \cref{eq:approx:detA-detB})}
\end{aligned}
\end{align}
This completes the proof.
\end{proof}

\section{Open Problems}\label{sec:conclusion}
We investigated the \W{1}-hardness of \prb{Determinant Maximization} in the three restricted cases, improving upon the result due to \citet{koutis2006parameterized}.
Our parameterized hardness results leave a few natural open problems:
For what kinds of sparse matrices is \prb{Determinant Maximization} FPT?
Is there a $(1-\epsilon)$-factor (rather than ``additive'') FPT-approximation algorithm with respect to the matrix rank?
Quantitative lower bounds can be also proved;
e.g., due to the lower bound of \prb{$k$-Sum} \cite{patrascu2010possibility},
\prb{Determinant Maximization} on tridiagonal matrices
cannot be solved in $n^{o(k)}$ time,
unless \emph{Exponential Time Hypothesis} \cite{impagliazzo2001complexity,impagliazzo2001problems} fails.

\section*{Acknowledgments}
I would like to thank Jon Lee for pointing out to us the reference due to \citet{al-thani2021tridiagonal,al-thani2021tridiagonala}.
I thank the anonymous referees for their comments
which made the presentation of this paper much better, and
for suggesting a simple procedure for generating primitive Pythagorean triples.

\printbibliography

\end{document}